\documentclass{aims} 
\usepackage{amsmath}
\usepackage{paralist}
\usepackage{epsfig} 
\usepackage{graphicx}  
\usepackage{epstopdf} 
\usepackage[colorlinks=true]{hyperref}
\hypersetup{urlcolor=blue, citecolor=red}
\allowdisplaybreaks

\def\currentvolume{X}
 \def\currentissue{X}
  \def\currentyear{200X}
   \def\currentmonth{XX}
    \def\ppages{X--XX}

\newtheorem{theorem}{Theorem}[section]

\theoremstyle{definition}
\newtheorem{definition}[theorem]{Definition}
\newtheorem{remark}{Remark}

\newcommand{\mres}{\mathbin{\vrule height 1.6ex depth 0pt width
0.13ex\vrule height 0.13ex depth 0pt width 1.3ex}}

\usepackage{amsfonts,amssymb}
\usepackage{stmaryrd}
\usepackage{mathrsfs}
\usepackage{framed}
 \usepackage{subcaption}
\title[An SBV framework for stripe patterns] 
{An SBV relaxation of the Cross-Newell energy for modeling stripe patterns and their defects} 

\author[Shankar C. Venkataramani]{}

\subjclass{Primary: 35B36, 49J45; Secondary: 26A45, 90C46.}
 \keywords{Pattern formation, phase reduction, measured foliations, special functions of bounded variation, relaxation, split-Bregman method.}

 \email{shankar@math.arizona.edu}
 
\thanks{The author is supported by NSF grant GCR-20202915}

\begin{document}
\maketitle

\centerline{\scshape Shankar C. Venkataramani}
\medskip
{\footnotesize
 \centerline{Department of Mathematics, University of Arizona}
   \centerline{Tucson, AZ 85721}
   \centerline{USA}
} 

\bigskip



\begin{abstract}
We investigate stripe patterns formation far from threshold using a combination of topological, analytic, and numerical methods. We first give a definition of the mathematical structure of `multi-valued' phase functions that are needed for describing layered structures or stripe patterns containing defects. This definition yields insight into the appropriate `gauge symmetries' of patterns, and leads to the formulation of variational problems, in the class of special functions with bounded variation, to model patterns with defects. We then discuss approaches to discretize and numerically solve these variational problems. These energy minimizing solutions support defects having the same character as seen in experiments.
\end{abstract}


\section{Introduction} \label{sec:intro}

A common motif seen in self-organized systems is that of {\em stripe patterns} or  {\em layered structures}. We see this motif in the epidermal ridges on ones fingers as well as in the accordion like patterns of spines and valleys on a saguaro cactus (see Fig.~\ref{cactus}). Similar patterns are also observed in laboratory experiments on convection (See Fig.~\ref{ellipse-expt}), liquid crystals, and the wrinkling of elastic thin films attached to a substrate. 

An ideal stripe pattern is built from uniformly spaced parallel layers. Consequently, stripe patterns have a discrete translation symmetry normal to the layers, and a continuous translation symmetry along the layers. This is a reflection of the generic mechanisms that gives rise to stripe patterns -- they are typically the result of a symmetry breaking bifurcation where a homogeneous (`featureless') ground state (with the ``full" continuous translation symmetry group) is supplanted by a symmetry broken state with a discrete translation symmetry. 

\begin{figure}[htbp]
\centering
\begin{subfigure}[b]{0.4 \textwidth}
 \includegraphics[height = 0.9 \textwidth]{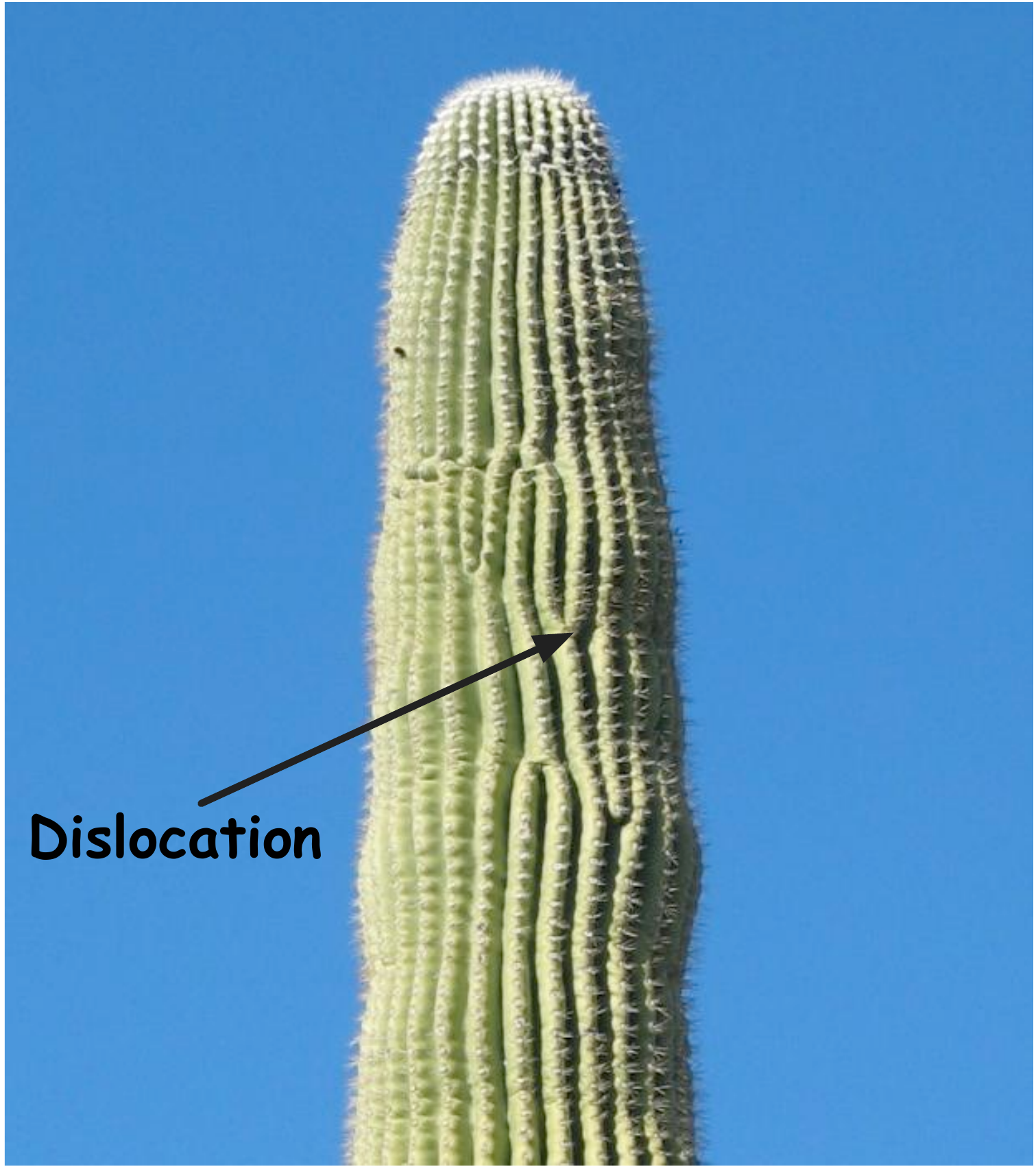} 
 \caption{Patterns and defects in Cacti.}
 \label{cactus}
 \end{subfigure}
\begin{subfigure}[b]{0.58 \textwidth}
 \includegraphics[width = 0.9 \textwidth]{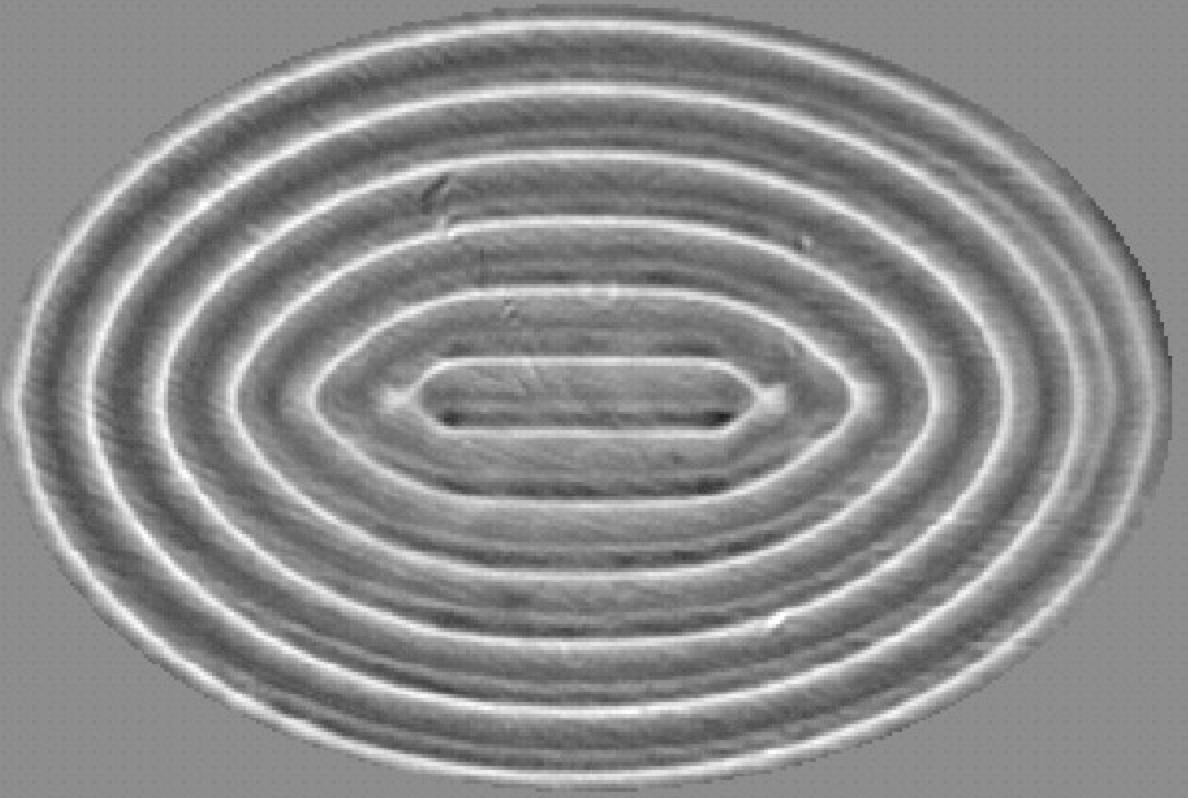} 
 \caption{Convection is an elliptical container.}
 \label{ellipse-expt}
 \end{subfigure}
\caption[Stripe Patterns and point defects]{\label{fig:motivate} (a) An accordion pattern of rows of spines and intervening valleys on {\em Carnegiea gigantea} -- the Saguaro cactus. The folds allow the cactus to expand and store water when available. A bifurcation defect where a layer of spines splits into two is indicated. (b) Convection pattern in an elliptical container with gently heated sidewalls. Heating the sidewalls aligns the direction of the stripes with the boundary corresponding to Dirichlet boundary conditions for the `phase'. Image reproduced with permission from Ref.~\cite{Meevasana2002convection}.}
\end{figure}

 Observed stripe patterns are seldom ideal. Rather, layers can bend and the spacing between layers can be nonuniform corresponding, respectively, to {\em bending} and {\em stretching} deformations of an ideal pattern. Stretching and bending deformations are `topologically nice' and indeed there are useful analogies with bending and stretching deformations of this elastic sheets, that have been explored elsewhere \cite{newell2017elastic}. In the context of this work, stretching and bending deformations break the global symmetries of an ideal stripe pattern, but the symmetry under translations along the stripes is still manifested locally. The focus of this work is stripe patterns with {\em defects}, i.e. topological features that serve as obstructions to `unbending' and `unstretching' a given pattern into an ideal stripe pattern.  In `nice' deformations of the ideal stripe pattern, layers can bend and the separation can become non-uniform, but stripes never merge and they can only ``end" on the boundary of the domain. In contrast, we generically observe defects, including the ending of a layer, or the bifurcation of one layer into two, in naturally occurring stripe patterns, as illustrated in Fig.~\ref{fig:motivate}. In the context of epidermal ridges, these features, called {\em minutiae}, are key to the usefulness of fingerprints for forensics \cite{ratha2000robust}. 
 
In this paper, we restrict our attention to energy driven patterns on planar (two-dimensional) domains. We consider physical systems  in which the first bifurcation is from a homogeneous state to a striped pattern which has only a discrete periodic symmetry in one direction. Further, we retain the full orientational symmetry of the homogeneous state, so there is no preferred orientation for the striped pattern. The \emph{symmetry-breaking} bifurcation to stripes occurs at a critical threshold; above this threshold the pattern can deform and, further away, \emph{defects} can form. Our goal is to formulate models, as (generalized) gradient flows, for the dynamical processes involved in the birth and the motion of defects.

A representative example physical systems that can be described as stripe patterns with defects is high Prandtl number Rayleigh-B\'enard convection. The conduction state is the initial homogeneous state and at a critical Rayleigh number fluid convection is initiated, leading to a ``stripe pattern", which  can be seen in the vertical velocity or equivalently the temperature field on horizontal cross-section across the middle of the experimental cell. This scenario is illustrated in Fig.~\ref{fig:convection} and an illustration the resulting patterns and its defects is shown in Fig.~\ref{shadowgraph}.

\begin{figure}[htbp]
\centering
\begin{subfigure}[b]{0.53 \textwidth}
 \includegraphics[height = 0.9 \textwidth]{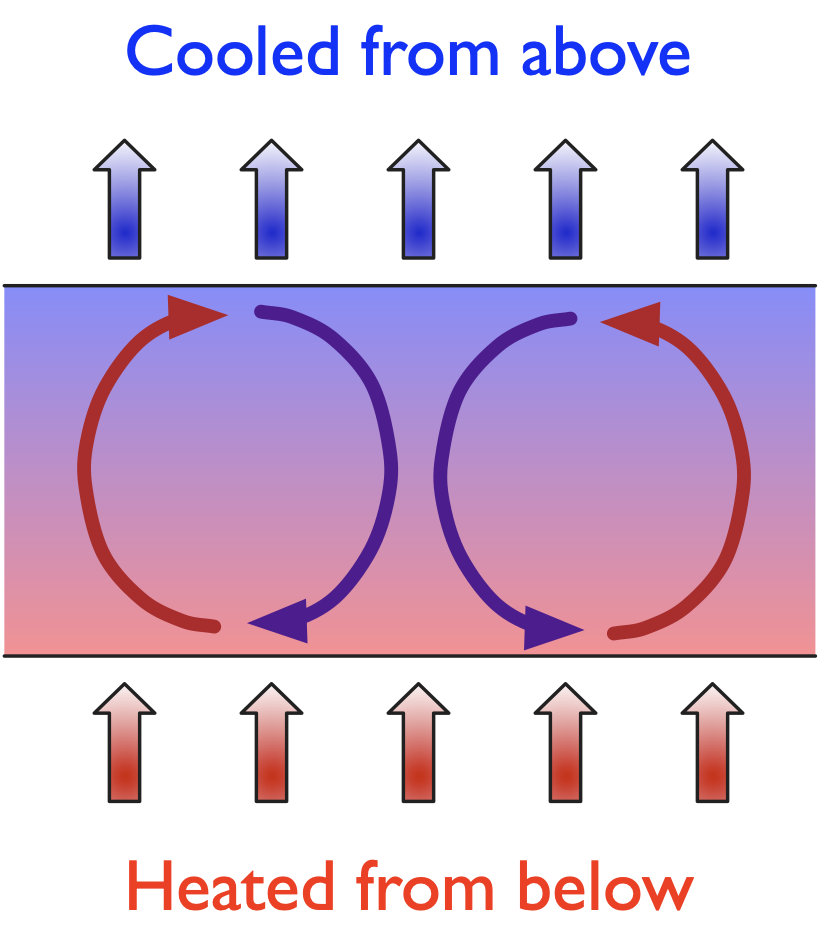} 
 \caption{Schematic of a convection cell.}
 \label{fig:convection}
 \end{subfigure}
\begin{subfigure}[b]{0.45 \textwidth}
 \includegraphics[width = 0.9 \textwidth]{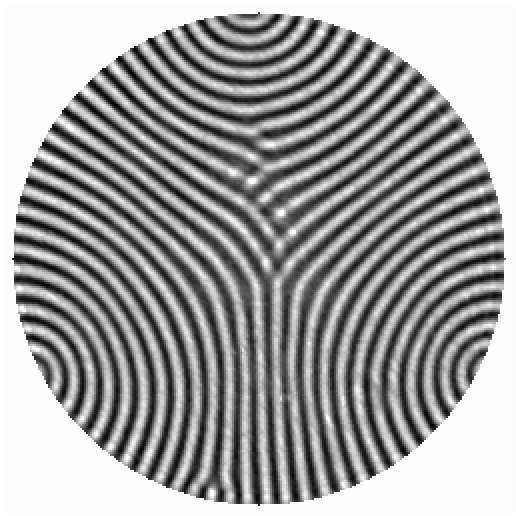} 
 \caption{Shadowgraph of a convection pattern displaying point defects.}
 \label{shadowgraph}
 \end{subfigure}
\caption[Stripe Patterns and point defects]{\label{fig:fluid-cells} (a) A schematic of the convection process. Hot fluid rises from the heated bottom and displaces the cold fluid at the top by buoyancy, setting up a convection roll. (b) A Rayleigh-B\'enard convection pattern. Image reproduced with permission from Ref.~\protect{\cite{Liu1996Spiral}}.}
\end{figure}

In regions that are free of defects, stripe patterns (locally) manifest a continuous translation symmetry along layers and a discrete translation symmetry normal to the layers. Ideal patterns, which manifest these symmetries globally, are described by a periodic function of a \emph{phase}, $\theta = k\cdot x + \theta_0$, where the magnitude $|k|$ is the wavenumber of the pattern and the orientation of $k = (k_1,k_2)$ is perpendicular to the stripes. Here $x=(x_1,x_2)$ is a physical point in the plane. Even though the stripe pattern will deform far from threshold, over most of the field (and in particular away from defects) it can be locally approximated as a function of a  phase $\theta(x)$, for which a local wavevector can be defined as $k = \nabla\theta$ which, as a consequence of the (local) discrete translation symmetry normal to the stripes, is nearly a constant vector over distances on the order of multiple inter-layer separations. 

\subsection{The Swift-Hohenberg equation}
A generic model for the formation of stripe patterns is the Swift-Hohenberg equation \cite{swift1977hydrodynamic}
\begin{equation}
    \frac{\partial \psi}{\partial t} = (R - (1 + \nabla^2)^2) \psi - \psi^3,
    \label{eq:sh}
\end{equation}
which is the $L^2$ gradient flow for the  energy functional
\begin{equation}
    \label{eq:sh_energy}
    \mathcal{F} = \int_\Omega \frac{[(1 + \nabla^2) \psi]^2}{2} + \frac{\psi^4}{4} -  \frac{R \psi^2}{2},
\end{equation}
where $\Omega \subseteq \mathbb{R}^2$ is a planar domain and $\psi:\Omega \to \mathbb{R}$ is a real field. The energy functional is invariant under $\psi \rightarrow - \psi$. Consequently, the homogeneous state $\psi = 0$ is always a steady solution. From the translation invariance of the energy, we can characterize the linearization about this steady state using the Fourier modes $\psi = \delta a_{k}(t) e^{i k \cdot x}$ with $\delta \ll 1$. Dropping the nonlinear terms we obtain
$$
\dot{a}_{k} = \left(R - (1-|k|^2)^2 \right) a_{k}
$$
The homogeneous state is consequently unstable to the formation of rolls for $R > 0$ and $|k|$ in an interval $(k_-,k_+)$ containing $|k|=1$. For every spatially constant wave-vector $k$ with $|k| \in (k_-,k_+)$, there exists a one parameter family of stable critical points for $\mathcal{F}$ of the form $\psi(x) = w_0(k \cdot x + \theta_0,|k|^2)$, where $w_0$ is a $2\pi$-periodic function of it's first argument \cite{CEbook}, with 
\begin{equation}
w_0(\theta + \pi,|k|^2) = - w_0(\theta,|k|^2), \quad \quad w_0(-\theta,|k|^2) = w_0(\theta,|k|^2).
\label{eq:F-symm}
\end{equation}
The linearization about these solutions has a neutral mode, corresponding to translations in space, or equivalently in the phase $k \cdot x \equiv \theta \to \theta+\delta$, and the linearization is stable to perturbations that are orthogonal to the neutral mode \cite{CEbook}.

Cross and Newell \cite{cross1984convection} argue that, away from defects, stripe patterns are modulations of these stable critical points,
\begin{equation}
\psi(x_1,x_2,t) = w_0(\theta(x_1,x_2,t),|\nabla \theta|^2) + O(|\nabla k|),
\label{eq:modulation}
\end{equation}
where $k = \nabla \theta$ varies slowly in space and time, i.e $|\nabla k| \ll 1$. This is a powerful idea, and using the modulation ansatz with the method of multiple scales, Cross and Newell obtain a phase diffusion equation that describes the long wavelength behavior of stripe patterns, even far from onset \cite{cross1984convection}. We will return to this point in Sec.~\ref{sec:rcn}.

While a smooth and slowly varying field $\theta(x_1,x_2,t)$ gives a smooth stripe pattern from the modulation equation, the converse is not true. 
This is due to the structure of the space on which the phase $\theta(x_1,x_2)$ ``lives". In particular, the phase $\theta$ is not physically observable, unlike the order parameter $\psi \approx w_0(\theta,|\nabla\theta|^2)$, which represents, for example, the vertically averaged temperature in a convection pattern. We thus have to identify (local representatives) of different phase functions $\theta(x_1,x_2)$ that give the same order parameter $\psi = w_0(\theta(x_1,x_2))$. From the symmetries in Eq.~\eqref{eq:F-symm}, it follows that the phase $\theta$ is therefore identified with $\theta + 2 n \pi, n \in \mathbb{Z}$ (periodicity) and also with $ - \theta$ (even function) 
as illustrated in fig.~\ref{fig:phase}.  As a consequence, even in regions where the modulation ansatz~\eqref{eq:modulation} is valid, the phase is  a ``multi-valued" function \cite{newell1996defects,ercolani2000geometry}. We will refer to the allowed transformations of the phase function, that keep the `physical field' $\psi$ invariant, as {\em gauge transformations}.

\begin{figure}[t]
\vspace*{-0.5pc}
\centering
\begin{subfigure}[b]{0.36 \textwidth}
 \includegraphics[width = \textwidth]{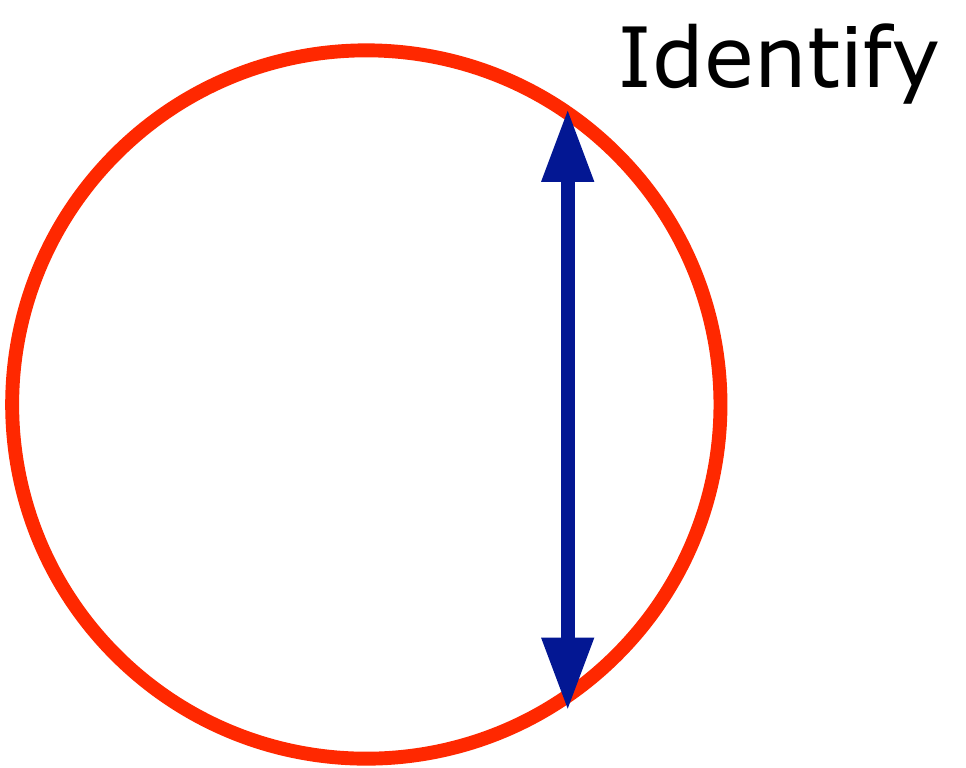} \\
  \includegraphics[width = 0.75 \textwidth]{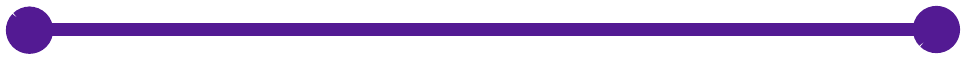}
  \caption{``Phase" space.}
  \label{fig:phase}
 \end{subfigure}
  \begin{subfigure}[b]{0.3\textwidth}
        \includegraphics[width=0.9\textwidth]{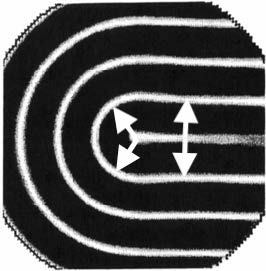}
        \caption{Convex disclination.}
        \label{convex}
    \end{subfigure}
    \begin{subfigure}[b]{0.3\textwidth}
        \includegraphics[width=0.9\textwidth]{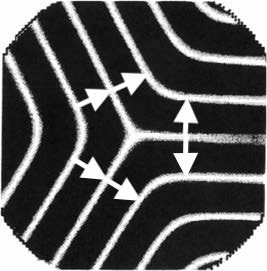}
        \caption{Concave disclination.}
        \label{concave}
    \end{subfigure}
\caption[Disclinations]{\label{fig:disclination} The geometry of ``Phase" space and associated point defects. }
\end{figure}

The approach in this work builds on ideas that were introduced earlier \cite{ercolani2009variational,Zhang2021Computing}. In particular, it is crucial that our models do not have additional symmetries beyond those of the underlying physical system \cite{ercolani2009variational}, and that our computational approach is able to capture the non-orientablity of patterns due to the presence of point defects \cite{Zhang2021Computing}. In this work, we will find it useful to consider the `extended phase space' comprising of the phase $\theta(x_1,x_2)$ and the corresponding wavevector $k = \nabla \theta$. The symmetries of $w_0$ imply the `extended' symmetries $\theta \to \theta + 2n \pi, k \to k$ and $\theta \to - \theta, k \to - k$. One can define topological invariants, based on the monodromy of $(\theta,k)$ on traversing a closed circuit in the domain \cite{Zhang2021Computing}. In particular, any circuit that results in a flip of $k$ encloses a net {\em disclination} and the `elementary' disclinations with degree $\pm \frac{1}{2}$ are respectively the convex and concave disclinations illustrated in Fig.~\ref{convex}~and~\ref{concave}. 

\subsection{An outline with the main results}

The goal of this work is to develop numerical methods to investigate the evolution of stripe patterns with particular attention to the role of point defects. This requires us to bring together ideas pertaining to topological defects \cite{mermin1979topological}, free discontinuity problems and special functions of bounded variation (SBV) \cite{degiorgi1988sbv}, and numerical methods that meld ideas from convex optimization and compressed sensing \cite{goldstein2009}. This paper is organized into sections, that roughly reflect this splitting of ideas/methods.

In \S\ref{sec:foliations} we begin by reviewing the topological notion of measured foliations \cite{Thurston1988Geometry}, which give the natural setting for studying layered media \cite{poenaru1981some}. We argue that layered structures are represented by a restricted class of measured foliations, those that have an unambiguous definition of layers and half-layers \cite{Machon2019Aspects}. We propose a mathematical condition that would enforce this restriction (See Eq.~\eqref{eq:overlap} below). Then we discuss a `gauge symmetry' that is automatically implied by a phase description of stripe patterns, and the quest to pick a canonical representative among the various gauge equivalent phase descriptions of a given stripe pattern. We frame this discussion in the context of the {\em orientability} of a stripe pattern \cite{Knoeppel2015Stripe}. Our key result is Thm.~\ref{thm:orient} that allows us to pick a representative that minimizes an appropriate `defect' and is thus natural from the viewpoint of energy minimization.

In \S\ref{sec:rcn} we discuss the how one can formulate  variational problems that allow for the possibility of non-orientable defects without any apriori restrictions on the nature of the defect set. This procedure is agnostic about the underlying physical model, and can be adapted to a variety of energy functionals that are used to describe convection patterns \cite{cross1984convection} or smectic liquid crystals \cite{aviles1987mathematical,Machon2019Aspects}.  We also review the theory of SBV functions and their applications to free discontinuity problems \cite{degiorgi1988sbv}. The two key results in this section are (i) the identification of the appropriate class of admissible SBV functions that correspond to stripe patterns (See Eq.~\eqref{admissible_class}), an analytic formulation of the appropriate topological condition (See Eq.~\eqref{eq:overlap}), and (ii) a relaxation of the pattern energy functional to an unconstrained variational problem on SBV, such that the minimizer of the relaxed problem also gives a (constrained) minimizer in the admissible class $\mathcal{A}$.

In \S\ref{sec:bregman} we formulate a numerical method to study the variational problems on SBV spaces that come up in our modeling of stripe patterns. To this end, we combine ideas from numerical convex optimization \cite{yin2008bregman}, the method of minimizing movements \cite{Ambrosio1995Minimizing} and convex splitting \cite{Ascher1995ImplicitexplicitMF}. Our method is discretized using finite elements and it allows us to time-step the abstract gradient flow $u_t = - \frac{\delta \mathcal{E}}{\delta u}$ even in situations where $\mathcal{E}$ is not Fr\'{e}chet differentiable, or even convex. We illustrate the method with explicit computations. We conclude in \S\ref{sec:discussion} with a short discussion of the questions that naturally arise from this work and its potential future extensions.

\section{Foliations and layered media}
\label{sec:foliations}

We now discuss the mathematical framework for describing stripe patterns. There is a natural identification between stripe patterns and textures in smectic liquid crystals \cite{Machon2019Aspects} that we will exploit in what follows. The appropriate mathematical structures for describing the topology of defects in smectic liquid crystals \cite{poenaru1981some} are {\em measured foliations}~\cite{Thurston1988Geometry,fathi2012thurston}, whose definition we now recall.

Let $\Omega$ be a simply connected open subset of $\mathbb{R}^2$ with  compact closure $\bar{\Omega}$ and a smooth boundary $\partial \Omega$. A smooth foliation of $\Omega$ is a disjoint  decomposition of $\Omega$ into leaves, i.e. smooth one dimensional manifolds. This notion can be made precise by demanding that $\Omega$ admit coordinates that map the leaves of the foliation into the `standard' foliation on $\mathbb{R}^2$, given by the collection of the `horizontal leaves' $\{(x_1,x_2) \, | \, x_2 = c\}$.
\begin{definition}
A (smooth) foliation $\mathscr{F}$ of $\Omega$ is a `local' product structure on $\Omega$, i.e. a collection of open sets $\{U_\alpha\}$ that cover $\Omega$, each equipped with a diffeomorphism $\varphi_\alpha:U_\alpha \to \mathbb{R}^2$, such that the transition functions $\varphi_\alpha\circ \varphi_\beta^{-1}:\varphi_\beta(U_\alpha \cap U_\beta)  \to \mathbb{R}^2$ are given by $(x_1,x_2) \mapsto (f_{\alpha \beta}(x_1,x_2), g_{\alpha \beta}(x_2))$, where $f_{\alpha \beta},g_{\alpha \beta}$ are smooth functions. 
\end{definition}

A natural idea is to identify the leaves is a foliation with the layers in a stripe pattern or a smectic liquid crystal. There are, however, a couple of issues to address in order to make an useful identification --
\begin{enumerate}
    \item Foliations, as defined above, are `topological' rather than `geometric' objects, i.e. we can make sense of continuous deformations of a foliation, but there isn't a natural notion of a distance between two leaves in a foliation. Stripe patterns, however, are characterized by a discrete translation symmetry normal to the layers {\em on a preferred length-scale}. We thus need to incorporate additional geometric information for foliations to serve as mathematical models of stripe patterns or smectic liquid crystals \cite{mermin1979topological,Machon2019Aspects}. 
    \item As defined above, every leaf in a foliation is a smooth one-dimensional manifold. Such `defect-free' foliations can therefore only represent continuous deformations of ideal patterns. For our purposes, we need to allow for the possibility of defects, i.e. points whose neighborhoods do not inherit a product structure from the local pattern. In particular, we want to allow for isolated convex and concave disclinations.
\end{enumerate}

These issues are addressed by the notion of measured foliations \cite{Thurston1988Geometry}. While we will only need to consider smooth domains $\Omega \subseteq \mathbb{R}^2$ with compact closure in this work, these ideas can be extended to general Riemann surfaces \cite{Hubbard1979Quadratic} and have applications to Teichmuller theory \cite{Strebel1984Quadratic}  and the geometry of surfaces \cite{Thurston1988Geometry,fathi2012thurston}.

\begin{definition} \label{def:foliation} Let $D=\{p_1,p_2,\ldots,p_k\}$ denote  a finite collection of points in a domain $\Omega$ with an associated set of non-negative integers $I = \{n_1,n_2,\ldots,n_k\}$. A measured foliation $\mathscr{F}$ of $\Omega$ with `defects' $p_i$ and corresponsing `indices' $n_i$ is defined by an open cover $\{U_\alpha\}$ of $\Omega \setminus D$, and diffeomorphisms $\varphi_\alpha:U_\alpha \to \mathbb{R}^2$ such that
\begin{enumerate}
    \item The transition functions on overlaps, $\varphi_\alpha\circ \varphi_\beta^{-1}:\varphi_\beta(U_\alpha \cap U_\beta)  \to \mathbb{R}^2$, are given by $(x_1,x_2) \mapsto (f_{\alpha \beta}(x_1,x_2), c_{\alpha \beta} +\epsilon_{\alpha \beta} x_2)$ where $c_{\alpha \beta} \in \mathbb{R}, \epsilon_{\alpha \beta} \in \{+1,-1\}$.
    \item The point $p_i$ is in the closure of $n_i$ (finitely many) leaves of the foliation $\{(U_\alpha,\varphi_\alpha)\}$ of $\Omega \setminus D$. 
\end{enumerate}
\end{definition}

A measured foliation induces a `transverse measure' on curves that do not contain the defects and are everywhere transverse to the leaves of the foliation. This measure is obtained by adding up the (absolute values of the) differences in the $x_2$ values in the coordinate charts along any curve that is transverse to the foliation \cite{Thurston1988Geometry}. This notion is well defined, as one can see from the overlap condition which guarantees that, for a pair of points $p,q$ that belong to two coordinate neighborhoods $U$ and $U'$, we have $|x_2(p)-x_2(q)| = |x_2'(p) - x_2'(q)|$.  This transverse measure is `geometric'. Comparing this transverse measure with the Hausdorff measure $\mathcal{H}^1$ allows us to naturally introduces a length scale to the foliation.

\begin{figure}[ht]
\centering
        \begin{subfigure}[t]{0.2\textwidth}
                \centering
                {\includegraphics[width=0.95\linewidth]{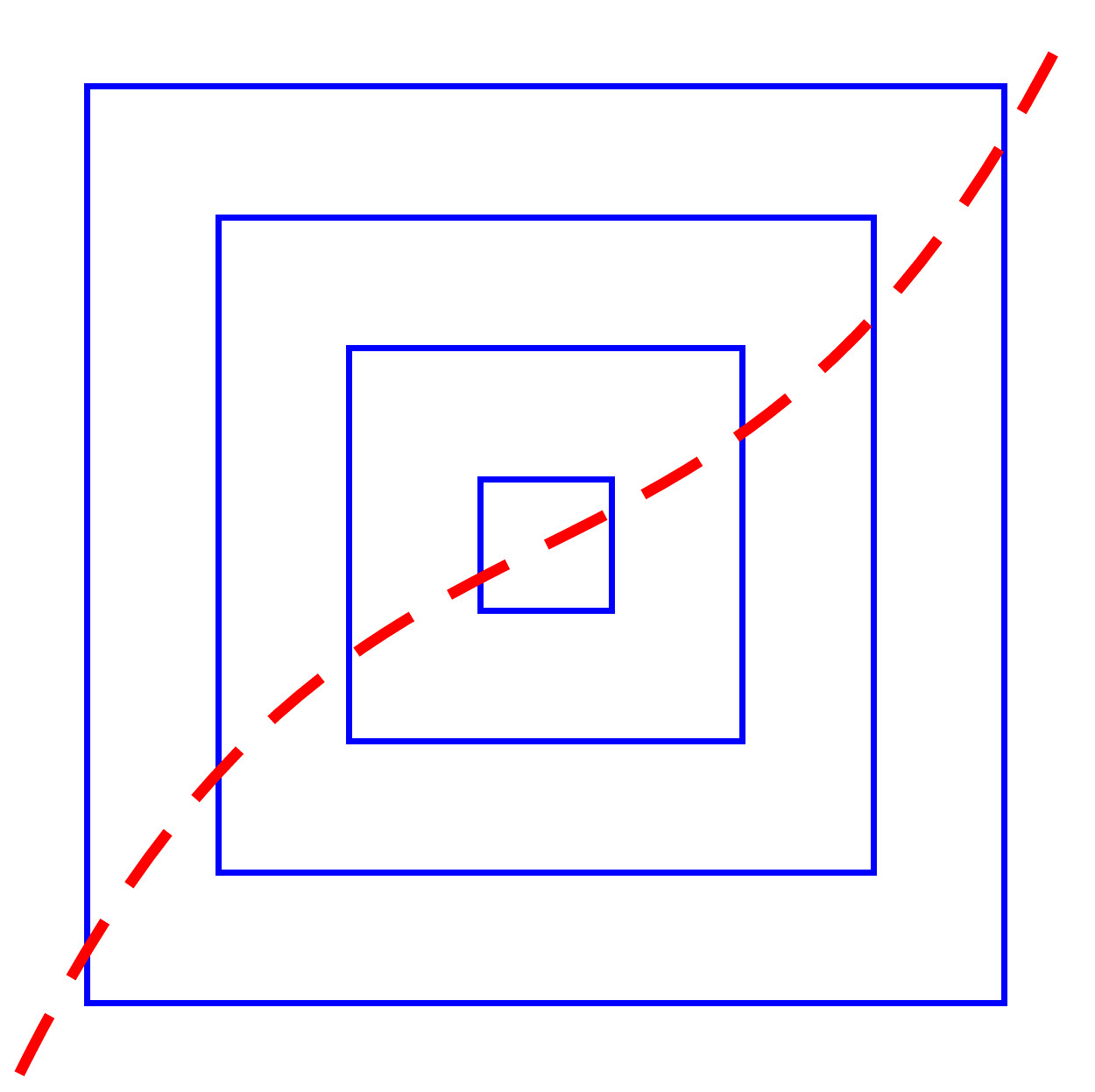}}
                \caption{}
                \label{fig:keq2}
        \end{subfigure}    
         \begin{subfigure}[t]{0.18\textwidth}
                \centering
                {\includegraphics[width=0.95\linewidth]{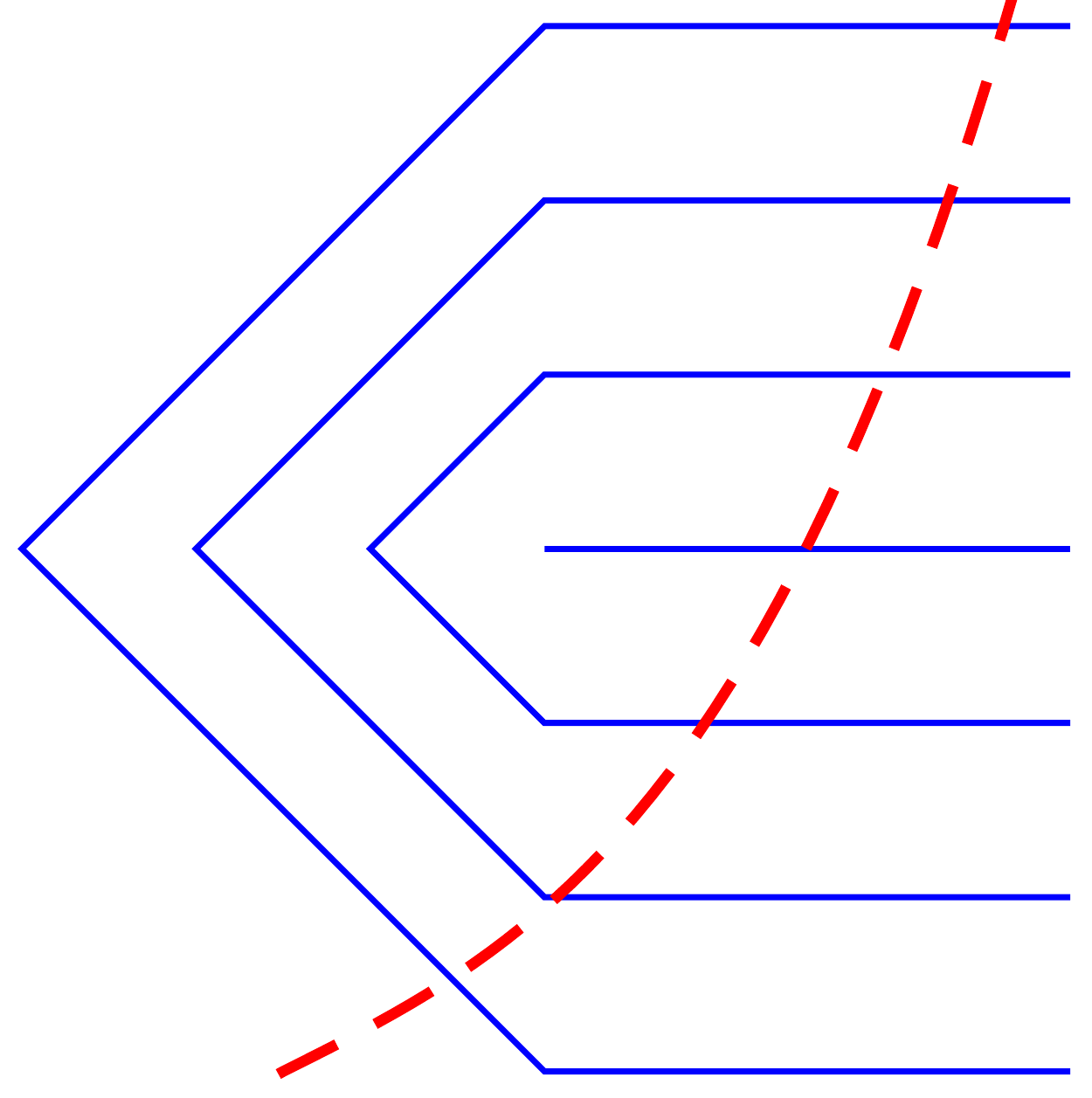}}
                \caption{}
                \label{fig:keq1}
        \end{subfigure}        
         \begin{subfigure}[t]{0.18\textwidth}
                \centering
                {\includegraphics[width=0.95\linewidth]{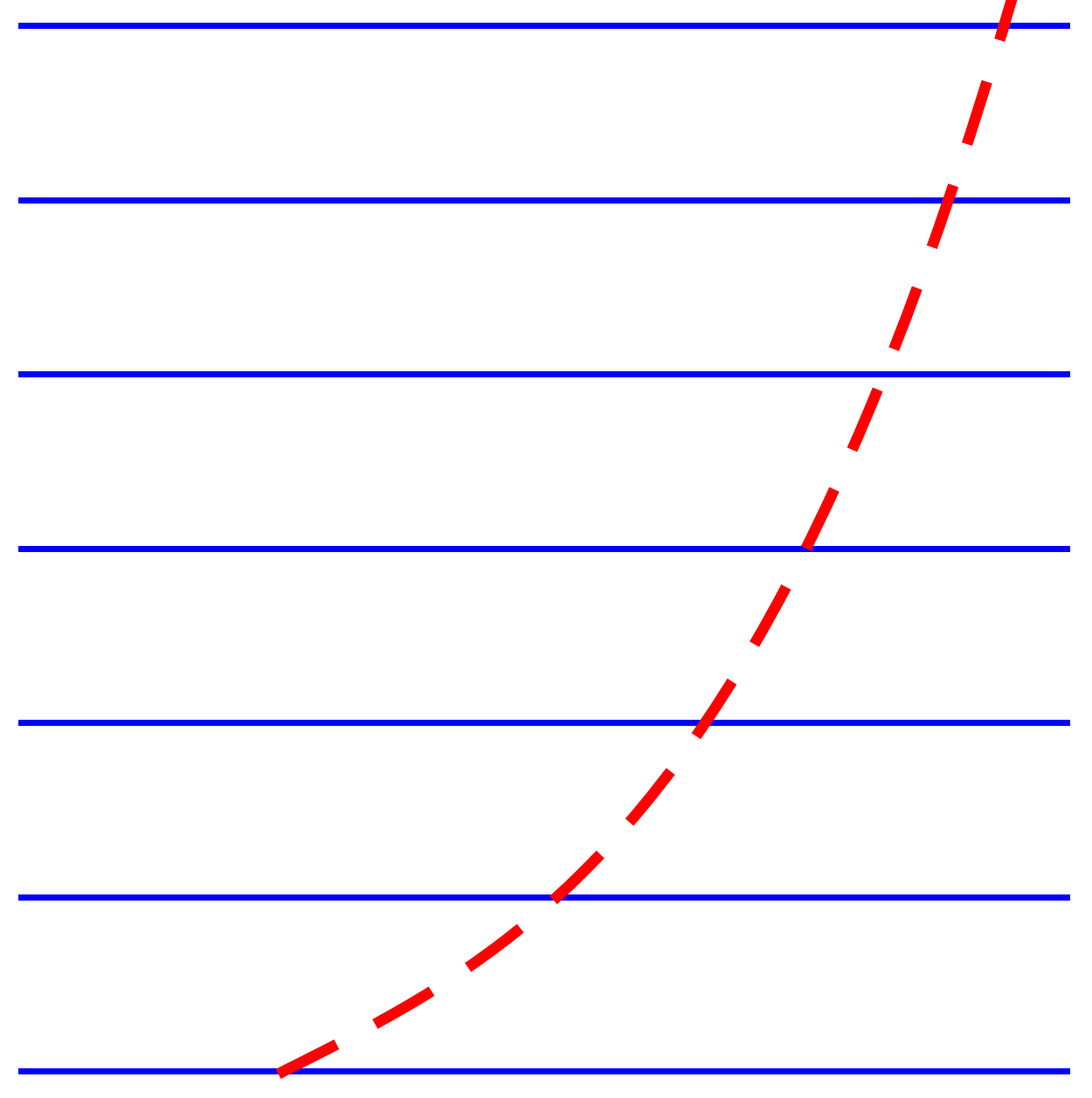}}
                \caption{}
                \label{fig:keq0}
        \end{subfigure}            
        \begin{subfigure}[t]{0.2\textwidth}
                \centering
                {\includegraphics[width=0.95\linewidth]{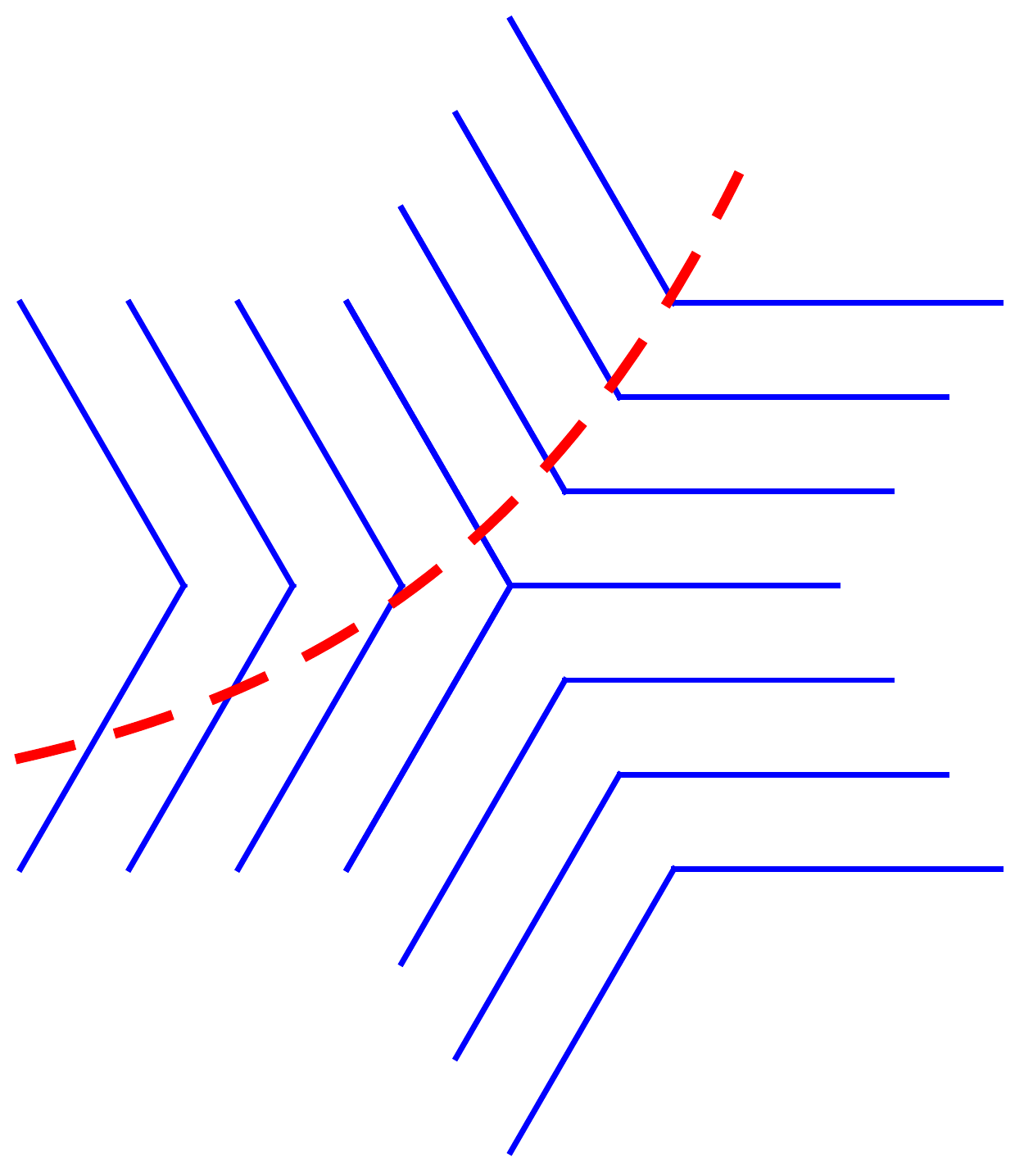}}
                \caption{}
                \label{fig:keqm1}
        \end{subfigure}    
         \begin{subfigure}[t]{0.18\textwidth}
                \centering
                {\includegraphics[width=0.95\linewidth]{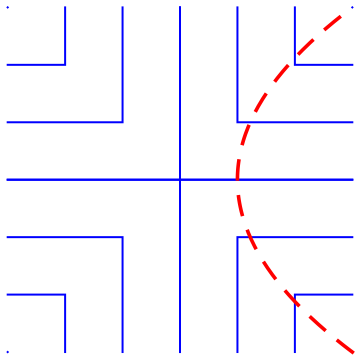}}
                \caption{}
                \label{fig:keqm2}
        \end{subfigure}        
        \caption{ Piecewise linear measured foliations depicting the local structure in the vicinity of isolated defects with indices, respectively, $n_i =0,1,2,3,4$. Note that the foliation in (c) with $n_i = 2$ is non-singular and the rest have an `essential' singular point. In (b)--(e) the dashed curve is transverse to the leaves in the foliation, while this is not the case for (a). In (a), the curve either contains the defect, or necessarily intersects a leaf in the foliation non-transversally. 
        }
        \label{fig:foliations}
\end{figure}

Examples of (piecewise linear) measured foliations on (subsets of) $\mathbb{R}^2$ are  shown in Fig.~\ref{fig:foliations}. The degree of a defect in a measured foliation is defined by the net rotation of a vector normal to the layers that is transported continuously along a small circle centered at the defect (see. Figs.~\ref{convex}~and~\ref{concave}). The degree $d_i$ is related to the index $n_i$ by $d_i = 1-\frac{n_i}{2}$ \cite{poenaru1981some}. The foliations in Figs.~\ref{fig:keq2}~and~\ref{fig:keq1} corresponding, respectively to a vortex and a convex disclination, therefore have degrees $+1$ and $\frac{1}{2}$ respectively. The foliations in Figs.~\ref{fig:keqm1}~and~\ref{fig:keqm2} correspond, respectively, to a concave disclination with $d=-\frac{1}{2}$ and a saddle with $d= -1$. 

\subsection{Phase fields} \label{sec:phases}

Measured foliations serve as models for `multi-valued' phase fields since we can identify the local phase $\theta$ with the second coordinate, i.e. $\varphi_\alpha(p) = (f_\alpha(p),\theta(p))$. The symmetries $\theta \to 2 n \pi \pm \theta$, that necessitate a multi-valued definition of the phase $\theta$, are encoded by the additional requirement that the overlap functions $\varphi_{\alpha}\circ \varphi_{\beta}^{-1}:\varphi_{\beta}(U_{\alpha} \cap U_{\beta})  \to \mathbb{R}^2$ should be of the form 
\begin{equation}
    (x_1,x_2) \mapsto (f_{\alpha \beta}(x_1,x_2), 2 k_{\alpha \beta} \pi +\epsilon_{\alpha \beta} x_2), \qquad k_{\alpha \beta} \in \mathbb{Z}, \epsilon_{\alpha \beta} \in \{+1,-1\}.
    \label{eq:overlap}
\end{equation}
We will henceforth assume this condition, which serves to define {\em quantized measured foliations}. Condition~\eqref{eq:overlap} is equivalent to the previously identified integrability condition for the director fields that describe layered media \cite[Eq.~(4)]{Machon2019Aspects}. 

Note that, as a consequence of condition~\eqref{eq:overlap}, the {\em height field} $h = \cos \theta$ is well defined on $\Omega \setminus D$. Indeed, the overlap condition gives, for $\theta:U \to \mathbb{R}, \theta': U' \to \mathbb{R}$, 
$$
h(p)= \cos \theta(p) = \cos (2 j  \pi \pm \theta'(p)) = \cos(\theta'(p)) = h'(p),
$$
for all points $p \in U \cap U'$ that are in the overlap of two coordinate neighborhoods. $\theta$ and consequently $h$ are constant on the leaves of the foliation $\mathscr{F}$ restricted to $\Omega \setminus D$. We will demand that $h$, as representing a ``physical'' field, should have a continuous extension to $\bar{\Omega}$, so that, in particular, $h$ is well defined at the defects $p_i$, while the `unphysical' and multi-valued phase $\theta$ {\em need not be as regular as} $h$. 

Since $\nabla h = -\sin \theta \, \nabla \theta$ and $\nabla\theta \to - \nabla \theta$ upon traversing a small circuit around a defect $p_i$ with an index $n_i$ that is odd, it is thus necessary that $\sin \theta =0 \Rightarrow \theta = m \pi$ on those leaves of the foliation $\mathscr{F}$ that contain (odd) half-degree defects $p_i$. This is illustrated in Fig.~\ref{fig:textures} which depicts examples of patterns/measured foliations including half degree-defects, namely convex and concave disclinations (See Fig.~\ref{fig:disclination}), which are on necessarily on the layers $\cos \theta = \pm 1$.

\begin{figure}[htbp]
\centering
\begin{subfigure}[htbp]{0.45 \textwidth}
\centering
 \includegraphics[height = 0.8 \textwidth]{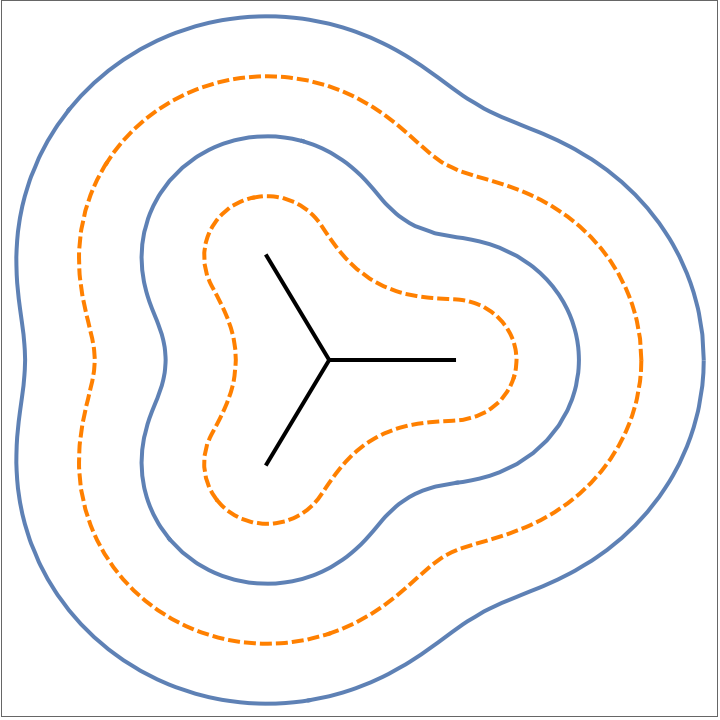} 
 \caption{ }
 \label{fig:3branch}
 \end{subfigure}
\begin{subfigure}[htbp]{0.45 \textwidth}
\centering
 \includegraphics[height = 0.8 \textwidth]{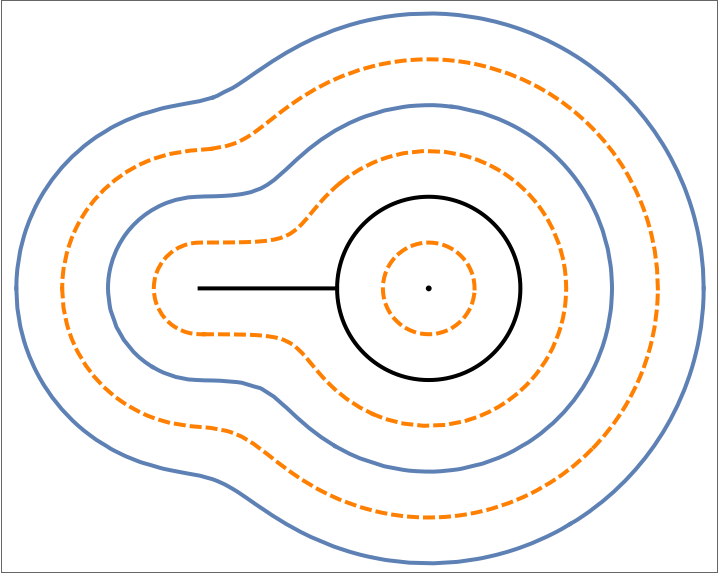} 
 \caption{ }
 \label{keyhole}
 \end{subfigure}
\caption[Stripe Pattern textures]{\label{fig:textures} Examples of possible stripe pattern textures with disclinations. We consider patterns satisfying a Dirichlet boundary condition $\theta = 0$ on $\partial \Omega$}
\end{figure}

Using the illustrative textures in Fig.~\ref{fig:textures}, we discuss the class of admissible mapings $\theta: \Omega \to \mathbb{R}$ that represent phase fields (See also Ref.~\cite{Machon2019Aspects}). We begin by introducing some notation and terminology. 
\begin{enumerate}
\item Let $L_t := h^{-1}(t)$ denote the level sets of the height function for $t \in [-1,1]$. 
\item The {\em texture} of a pattern is defined by the ordered pair $(L_{+1},L_{-1})$. We define $\Gamma_D := L_{+1} \cup L_{-1}$.
\end{enumerate}

By the previous argument, all the defects of the phase field have to live on $\Gamma_D$. and the leaves in the foliation restricted to $\Omega \setminus \Gamma_D$ are all non-singular. In this work, we will impose a Dirichlet boundary condition $\theta = 0$ on $\partial \Omega$. The level sets $L_t$ for $t  \in (-1,1)$ do not self-intersect (because they foliate) and  they also do not intersect the boundary $\partial \Omega$. Consequently, as illustrated in Fig.~\ref{fig:textures}, we have $\Omega \setminus \Gamma_D = h^{-1}((-1,1)) = \sqcup_i S_i$, a finite (disjoint) union of strips $S_i$, that are foliated by simple closed curves. Each strip is homeomorphic to the annulus $S^1 \times (-1,1)$.

We will require that the homeomorphism $\psi_i: S^1\times(-1,1) \to S_i$ extend to a Lipschitz mapping, also denoted by $\psi_i$, from $S^1 \times [-1,1] \to \bar{S_i}$, the topological closure of $S_i$. The boundaries of $S_i$ are then given by the images $\psi_i(S^1 \times \{1\}) \subseteq L_{+1}$ and $\psi_i(S^1 \times \{-1\}) \subseteq L_{-1}$. In particular, these boundaries are the images of Lipschitz maps, and are hence rectifiable and have a well defined tangent direction $\mathcal{H}^1$ a.e. Since the mappings $\psi_i$ are assumed Lipschitz on $S^1 \times [-1,1]$ it follows that the mappings $\psi_i(\cdot,t)$ from $S^1$ to $\bar{S_i}$ converge uniformly as $t \to \pm 1$. However, unlike the mappings for $t \in (-1,1)$, the ``boundary maps" for $t=\pm 1$ need not be one-to-one. This is illustrated in Fig.~\ref{fig:strips}, which show strips $S_i$ for the textures in Fig.~\ref{fig:textures}, with boundaries containing singular leaves of the corresponding foliations. 

We can label the strips in Fig.~\ref{fig:textures} so that $S_1$ is the ``outermost" strip and the index $i$ increases as we move inwards.  The texture in Fig.~\ref{fig:3branch} contains 4 strips labeled $S_1$ through $S_4$. The boundary of $S_4$ contains a $Y$ shaped `triod', i.e a union of three simple arcs that intersect at a common endpoint \cite{Alberti2013Structure}. As illustrated in Fig.~\ref{fig:a}, the triod is obtained as the uniform limit as $t \to 1$ of simple closed curves $L_t$. While the simple closed curves $L_t$ for $t \in (-1,1)$ can be oriented ``counterclockwise",  we cannot orient the triod. Points on the triod correspond to limits, as $t \to 1$, of well-separated points in $S^1$ that approach from ``either side" with opposite orientations.

The texture in  Fig.~\ref{keyhole} contains 6 strips, where the innermost strip $S_6$ is a punctured disk.  The boundary of $S_4$ contains a singular leaf as illustrated in Fig~\ref{fig:b}. The singular leaf $L_{+1}$ contains points where it can be consistently oriented (the circle) {\em and} points where the leaf cannot be oriented (the line segment).

\begin{figure}[htbp]
\centering
\begin{subfigure}[htbp]{0.45 \textwidth}
\centering
 \includegraphics[height = 0.8 \textwidth]{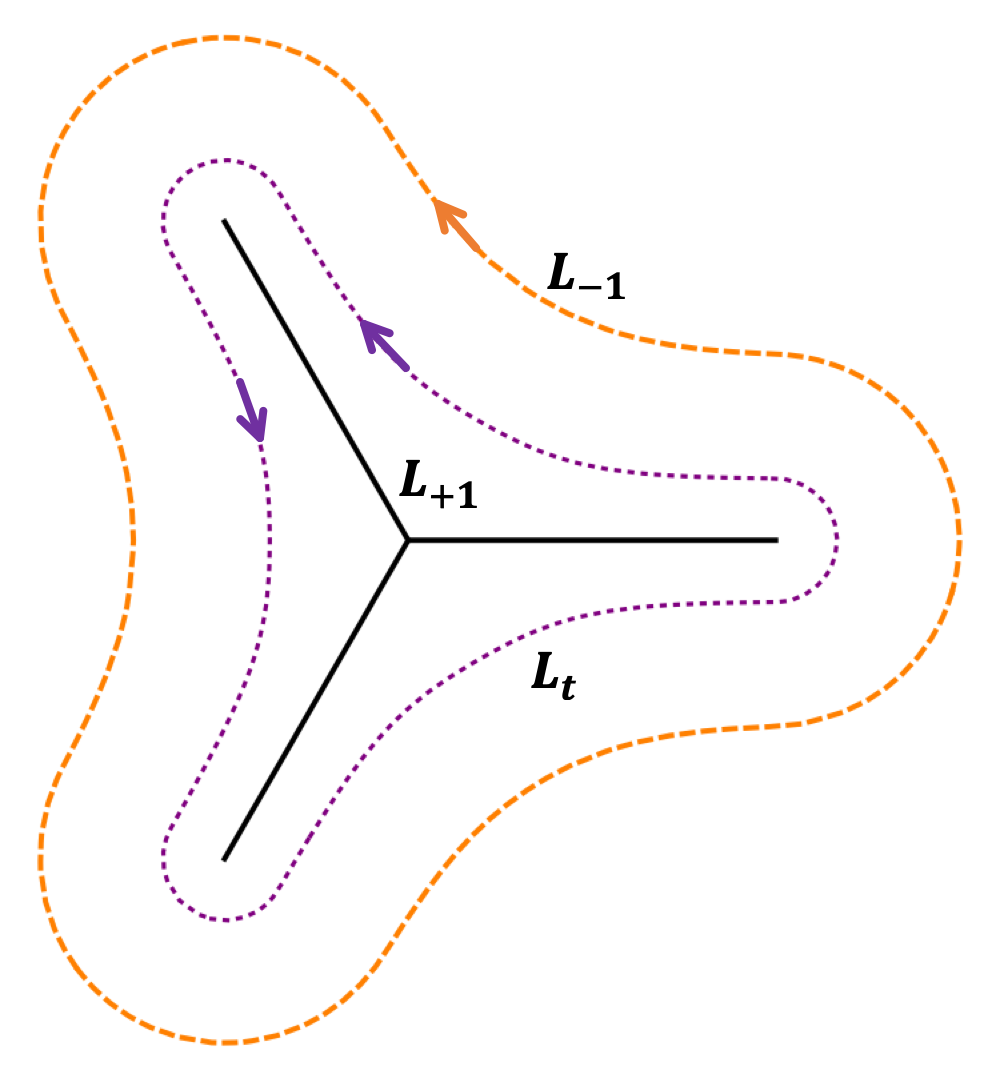} 
 \caption{ }
 \label{fig:a}
 \end{subfigure}
\begin{subfigure}[htbp]{0.45 \textwidth}
\centering
 \includegraphics[height = 0.8 \textwidth]{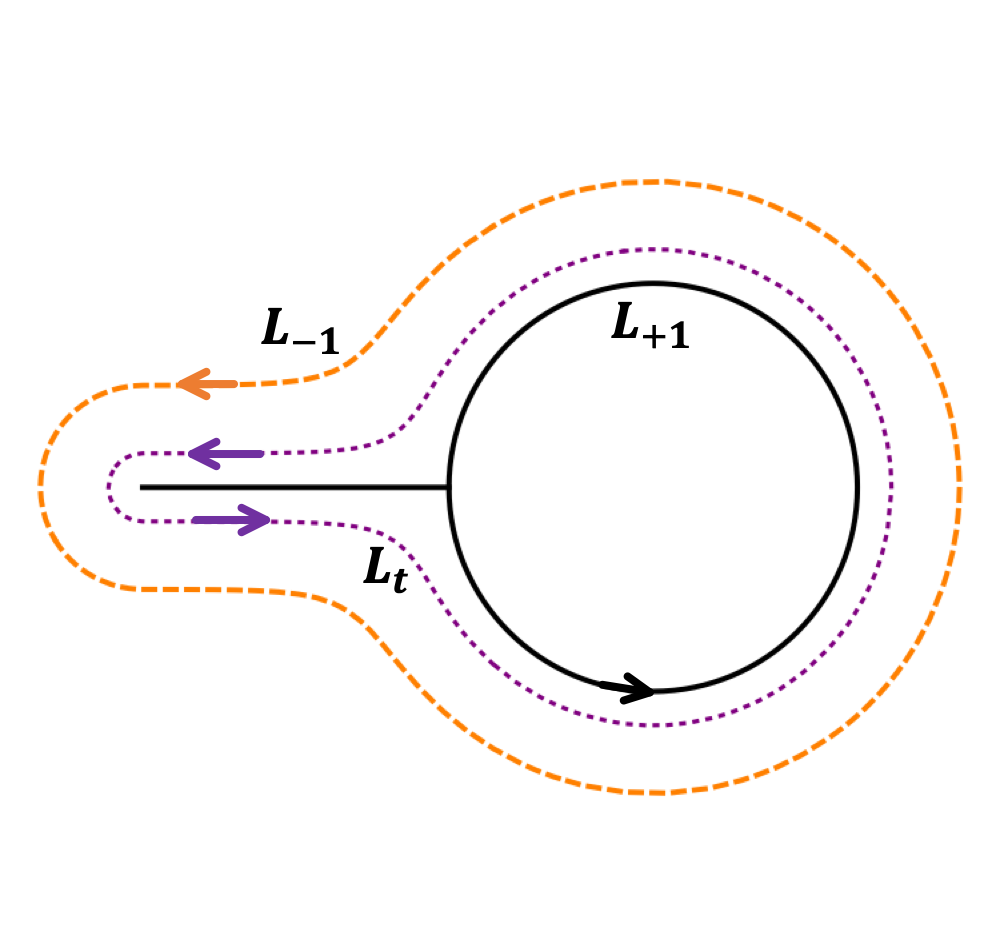} 
 \caption{ }
 \label{fig:b}
 \end{subfigure}
\caption[Stripe Pattern textures]{\label{fig:strips} Annular strips for the textures in Fig.~\protect{\ref{fig:textures}}. The figure illustrates strips whose boundaries contain singular leaves of the corresponding foliation. The figures also depict an oriented leaf $L_t$ for $t \in (-1,1)$ and illustrate the convergence of this `regular' leaf to the (potentially) singular leaves as $t \to \pm 1$. }
\end{figure}

A natural question is what determines is a point on the boundary of a strip is orientable, i.e the parameterization $S^1 \to L_{\pm1}$ is locally one-to-one or if it is not. To address this question we need to consider two distinct notions of the boundary of a set. We begin by elucidating this distinction, following De Giorgi  \cite{DeGiorigi1954teoria_misura}.

\begin{definition}
If $S \subseteq \mathbb{R}^2$ is an open set, the {\em topological boundary} of $S$ is defined by $\partial S = \bar{S} \setminus S$, where $\bar{S}$ is the closure of $S$ in $\mathbb{R}^2$.
If $S$ is a set of finite perimeter, the {\em reduced boundary} of the set, $\partial^*S$, is given by
$$
\int_{S} \mathrm{div}(\mathbf{v}) \,d\mathcal{L}^2 = \int_{\partial^*S} \mathbf{v} \cdot \mathbf{n} \,d\mathcal{H}^1,
$$
where $\mathbf{v}$ is any compactly supported smooth vector field on $\mathbb{R}^2$, $\mathcal{L}^2$ is two-dimensional Lebesgue (area) measure, and $\mathcal{H}^1$ is one-dimensional Hausdorff (length) measure. 
\end{definition}

We note a few properties of the reduced boundary $\partial^* S$.  If $S$ is a set with finite perimeter, then $\partial^* S$ is a $\mathcal{H}^1$ rectifiable set \cite{DeGiorigi1954teoria_misura}. Further, to within an $\mathcal{H}^1$ negligible set, $\partial^* S$ is a countable union of simple closed curves \cite{DeGiorigi1954teoria_misura}. $\partial^*S \subseteq \partial S$ and, in general, the reduced boundary $\partial^* S$ does not agree with the topological boundary $\partial S$. As in the examples above, a point $p \in \partial S$ is also in $\partial^* S$ only if the orientation on the contours $L_t$ for $t \in (-1,1)$ can be continuously extended to an orientation at $p$. In particular, this suggests that $\partial S \setminus \partial^* S$ consists of those boundary points to which we cannot continuously extend the orientation given to the leaves in $S$.

\subsection{Orientability of stripe patterns} \label{sec:orient}
By picking the principal branch $[0,\pi]$ for the range of $\arccos$, we can define a canonical representative for a phase function by $\tilde{\theta} = \arccos(h)$ where $h =\cos \theta$. The canonical representative $\tilde{\theta}$ is a uniquely defined single-valued object and it has the same regularity as $h$ away from $\Gamma_D$. This is the definition of the phase function that is used in Ref.~\cite{Machon2019Aspects}. 

Note however that, even if $\theta : \Omega \to \mathbb{R}$ is a smooth function, the canonical representative $\tilde{\theta}$ is only Lipschitz on neighborhoods of points $p$ where $\sin(\theta(p))=0, \nabla \theta(p) \neq 0$. This observation motivates the quest for a representative of a phase function that is ``maximally regular". One approach to this question is through the `orientability' of a stripe pattern \cite{Knoeppel2015Stripe}.

\begin{definition}
Let $\theta:\Omega \to \mathbb{R}^2$ be a globally Lipschitz phase function and let $\Gamma_D = \{(x_1,x_2) \,| \,\sin(\theta(x_1,x_2)) = 0\}$. Given a point $p \in \Omega \setminus \Gamma_D$, the level curve $\gamma$ of $\theta$ through $p$ is a Jordan curve with an `inner normal' $n_p$ that points into the bounded component of $\mathbb{R}^2\setminus \gamma$. The {\em induced orientation} at $p$ is defined by 
\begin{equation}
o_\theta(p) = \mathrm{sgn}(n_p \cdot \nabla \theta(p)).
\label{eq:orientation}
\end{equation}
\end{definition}

Equivalently, the induced orientation is given by $o_\theta(p) = dx_1\wedge dx_2\left(t_p,\frac{\nabla \theta}{|\nabla \theta|}\right)$ where $dx_1\wedge dx_2$ is the standard orientation on $\mathbb{R}^2$, $t_p$ is the unit tangent vector to a level curve of $\theta$ at $p$ with the standard (counterclockwise) orientation. Since the orientation $o_\theta$ is a smooth function away from $\Gamma_D$, and only takes the values $\pm 1$, it follows that $o_\theta$ is constant on each connected component of $\Omega \setminus \Gamma_D$, i.e. it each (open) strip $S_i$ has a constant orientation.

 If $\theta: \Omega \to \mathbb{R}$ is smooth with non-vanishing gradient, it is immediate that the induced orientation $o_{\theta}$ is a constant on all of $\Omega$. For the canonical representative $\tilde{\theta}$, however, $\nabla h$ changes direction, and consequently $\nabla \theta$ flips across every curve in $\Gamma_D$. The orientation $o_{\tilde{\theta}}$ for the canonical representative $\tilde{\theta}$ is therefore discontinuous at all the points on $\Gamma_D$ where the tangent direction $t_p$ is continuous. The induced orientation is thus not an invariant for the gauge symmetries that modify $\theta$ while preserving $h$.

 A natural measure, $\eta[\theta]$, of the size of the `orientation defect' is the perimeter of the boundary between the sets $o=1$ and $o=-1$, or equivalently, $\eta[\theta] = \frac{1}{2}\|o_\theta\|_{TV}$, the total variation of $o_\theta$. As we noted above, we can have $\cos(\theta_1) = \cos(\theta_2)$ but $\eta[\theta_1] \neq \eta[\theta_2]$. This motivates the question of how to redefine a given phase function $\theta$, using gauge symmetries, to eliminate/minimize the jumps in the induced orientation.  This question is answered by the following result, which asserts the existence of an {\em oriented representative} $\bar{\theta}$  for which $o_{\bar{\theta}} = 1$, i.e. $\eta[\bar{\theta}] = 0$ and there are no jumps in orientation.
 
Let $\theta$ be a phase function as discussed above. We begin by arguing that the orientation on each strip $S_i$ can be ``localized" independent of the orientations of the other strips $S_j$ with $j \neq i$.  Picking an (arbitrary) $t \in (-1,1)$, we have that $L_t \cap S_i$ is a Jordan curve in $\Omega$, as illustrated by the dotted curves in Fig.~\ref{fig:strips}. Since $\Omega$ is connected, it follows that $\Omega \setminus (L_t \cap S_i)$ has two connected components $A_{\pm}$ with $A_+$ containing $L_{+1} \cap \bar{S_i}$ and $A_-$ containing $L_{-1} \cap \bar{S_i}$. Let $\theta^{\pm}_i$ represent the values of the phase function $\theta$ on the (level) sets $L_{\pm1} \cap \bar{S_i}$ respectively.

Define the function $\psi_i: \bar{\Omega} \to \mathbb{R}$ by 
\begin{equation}
\psi_i(p) = \begin{cases} \theta(p) & p \in S_i \\  \theta^+_i & p \in A_+ \setminus S_i \\ \theta^-_i & p \in A_- \setminus S_i \end{cases}
\label{eq:potential}
\end{equation}
$\psi_i$ extends the Lipschitz function $\left. \theta \right|_{S_i}$ to a continuous function on $\bar{\Omega}$  that is constant on each connected component of $\bar{\Omega} \setminus S_i$ and is hence also a Lipschitz function. Further, we have 
\begin{equation}
\nabla \psi_i = \nabla \theta \cdot 1_{S_i}, \qquad \nabla(\cos(\psi_i)) = \nabla(\cos(\theta)) 1_{S_i},
\label{eq:localize}
\end{equation}
where $1_{S_i}$ denotes the indicator function for the set $S_i$.  Finally, we note that either $\psi_i = \theta^+_i$ or $\theta_i^-$ on $\partial \Omega$. In either case, $\psi_i$ is a constant on $\partial \Omega$.

\begin{theorem}
\label{thm:orient}
Let $\Omega \subset \mathbb{R}^2$ be a connected domain and let $\theta: \bar{\Omega} \to $ be a (Lipschitz) phase function satisfying a Dirichlet condition on $\partial \Omega$ with level curves that give a quantized measured foliation, as defined by Eq.~\eqref{eq:overlap}.  

Let $\Gamma_D = \{(x_1,x_2) | \sin(\theta(x_1,x_2)) = 0\}$ and let $S_i, i = 1,2,\ldots,k$ denote the finite collection of annular strips which constitute the connected components of $\Omega \setminus  \Gamma_D$. For every mapping $\epsilon: \{1,2,\ldots,k\} \to \{-1,1\}$, there is a phase function $\theta^\epsilon: \bar{\Omega} \to \mathbb{R}$ that is gauge-equivalent to the phase $\theta$,  that satisfies a Dirichlet boundary condition on $\partial \Omega$, and that has an orientation $o_{\theta^\epsilon}$ equal to $\epsilon(i)$ on the strip $S_i$.
\end{theorem}

\begin{proof} Let $\gamma: \{1,2,\ldots,k\} \to \{+1,-1\}$ denote the mapping defined by $\gamma(i)$ is equal to the orientation $o_\theta$ of the given phase function on the strip $S_i$, i.e. for all $p \in S_i$, we have $\mathrm{sgn}(n_p \cdot \nabla \theta(p)) = \gamma(i)$ (See Eq.~\eqref{eq:orientation}).

Define the function $\theta^\epsilon$ by
$$
\theta^\epsilon = \left[\sum_{i=1}^k \epsilon(i) \gamma(i) \psi_i \right] - c_0,
$$
where the functions $\psi_i$ are as defined in Eq.~\eqref{eq:potential} and $c_0$ is a constant that is determined by requiring $\theta^\epsilon =0$ on $\partial \Omega$. Consequently, $\cos(\theta^\epsilon) = \cos(\theta)$ on $\partial \Omega$.

It follows from Eqs.~\eqref{eq:orientation} and \eqref{eq:localize} that the orientation of $\theta^\epsilon$ on the strip $S_i$ is given by  $\mathrm{sgn}(n_p \cdot \nabla \theta^\epsilon(p)) = \mathrm{sgn}(\epsilon(i) \gamma(i)  (n_p \cdot \nabla \theta(p)) = \epsilon(i) [\gamma(i)]^2 = \epsilon(i)$. Also, since $\cos(\cdot)$ is an even function, it follows that $\nabla (\cos(\theta^\epsilon)) = \nabla(\cos(\theta))$ on $S_i$, {\em independent} of $\gamma(i)$ and $\epsilon(i)$. The functions 
$h = \cos(\theta), h^\epsilon = \cos(\theta^\epsilon)$ are both continuous on $\bar{\Omega}$ and have equal derivatives on $\Omega \setminus \Gamma_D$, a dense open set, so it follows that $\cos(\theta^\epsilon) = \cos(\theta)$ on $\bar{\Omega}$. This implies $\theta$ and $\theta^\epsilon$ are gauge equivalent.
 
\end{proof}

\begin{figure}[htbp]
\centering
\begin{subfigure}[htbp]{0.45 \textwidth}
\centering
 \includegraphics[height = 0.8 \textwidth]{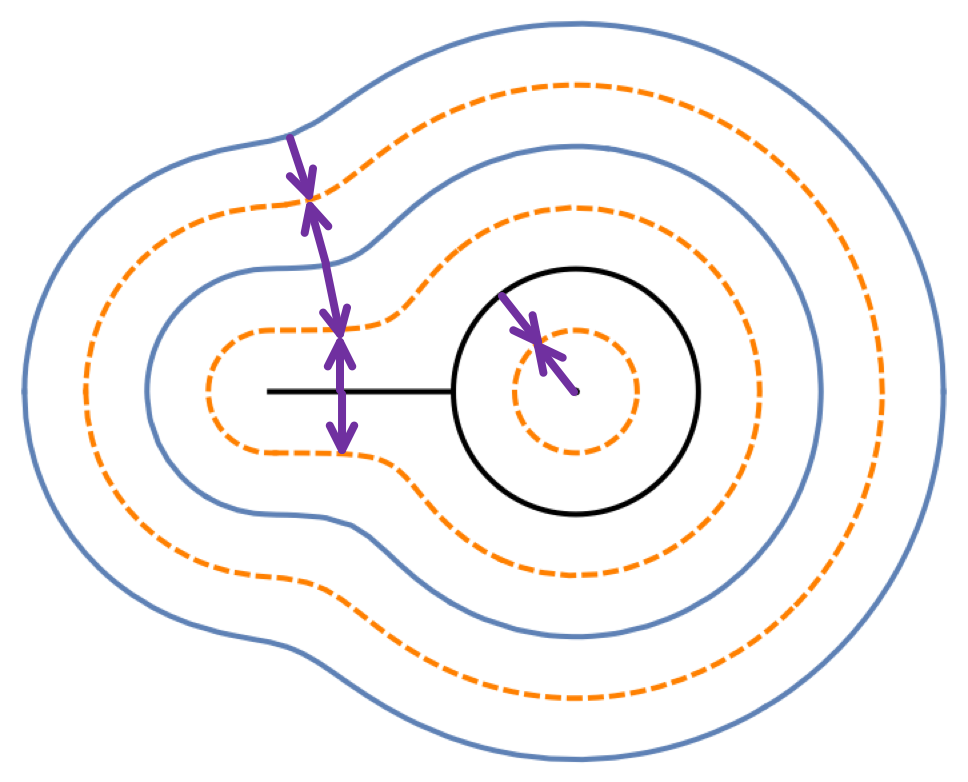} 
 \caption{The canonical phase field $\tilde{\theta}$. }
 \label{f:canonical}
 \end{subfigure}
\begin{subfigure}[htbp]{0.45 \textwidth}
\centering
 \includegraphics[height = 0.8 \textwidth]{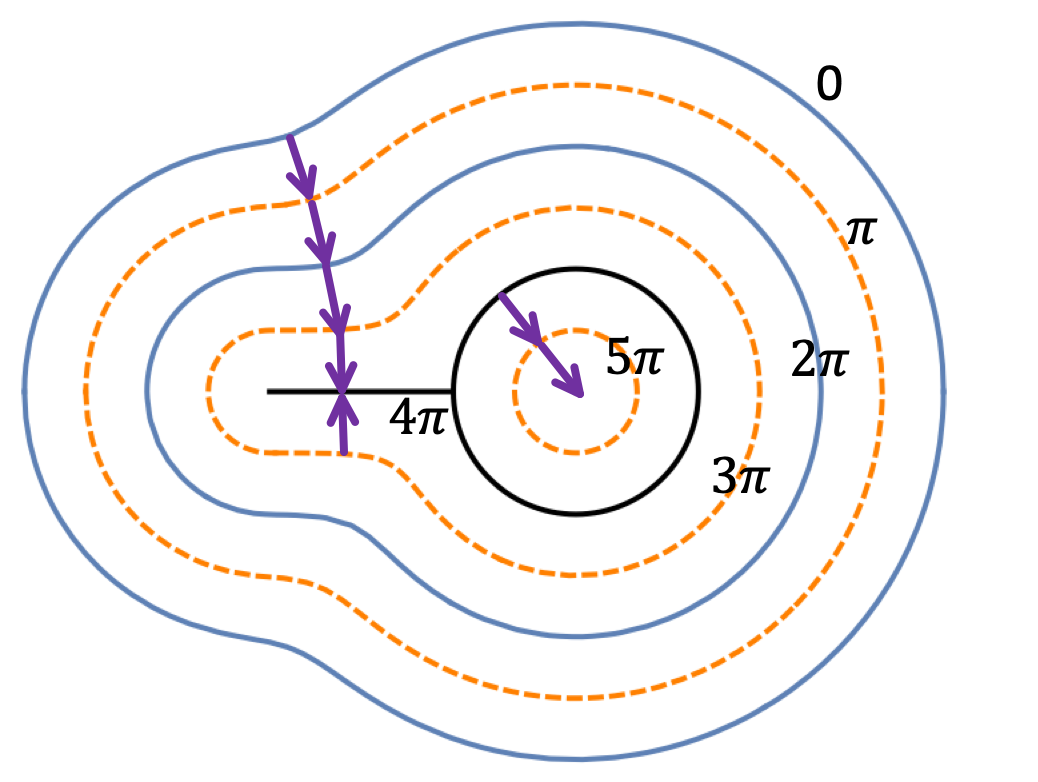} 
 \caption{The oriented phase field $\bar{\theta}$ }
 \label{f:oriented}
 \end{subfigure}
\caption[Stripe Pattern textures]{\label{fig:orientation} The canonical and the oriented phase fields for a given texture. (a) The solid contours correspond to $\tilde{\theta} = 0$, the dashed contours correspond to $\tilde{\theta} = \pi$, and $\nabla{\tilde{\theta}}$ jumps across every curve in $\Gamma_D$. (b) This figure illustrates the conclusion of Thm.~\protect{\ref{thm:orient}}. Using gauge symmetry, we can choose a phase function such that the induced orientation is always $+1$. The values of the phase at the various contours are indicated on the figure. }
\end{figure}

An immediate consequence of Thm.~\ref{thm:orient} is that, for any given texture, or equivalently, for any given phase function $\theta$, there is an equivalent phase function $\bar{\theta}$ that has a constant orientation on $\Omega$. For points in $\Gamma_D$ that are assigned a consistent orientation by the level sets $L_t$ from `either side', it follows that $n_p \cdot \nabla \bar{\theta}$ does not change sign, so there is no flip in the gradient of $\bar{\theta}$. For points in $\Gamma_D$ where the two sides give inconsistent orientations, the fact that $n_p \cdot \nabla \bar{\theta}$ has the same sign implies that $\nabla \bar{\theta}$ necessarily has to flip.  The advantage of picking the oriented representative $\bar{\theta}$ is that the set on which $\nabla \bar{\theta}$ jumps is as small as possible, and is associated with the disclinations in the underlying texture. 

This discussion is illustrated in Fig.~\ref{fig:orientation}. In particular, we want to emphasize the distinction between the induced orientation for a phase function, and the (non-) orientability of a point defect (See Fig.~\ref{fig:disclination}). For the phase field $\bar{\theta}$ shown in Fig.~\ref{f:oriented}, the induced orientation $o_{\bar{\theta}} =1$ on all the strips, while there is a jump in $\nabla \bar{\theta}$ across the line segment contained in the singular leaf $\bar{\theta} = 4 \pi$. This jump is unavoidable and is associated with the disclinations, i.e. the {\em non-orientable point defects} in the pattern.

\section{Energetics} \label{sec:rcn}

The starting point for a variational analysis of stripe patterns is the modulation ansatz in Eq.~\eqref{eq:modulation}. Introducing a small parameter $\epsilon \ll 1$, that sets the scale for $|\nabla k|$, we have $\psi \approx F_{k}(\epsilon^{-1}\Theta(\epsilon  x_1, \epsilon x_2))$, where $\epsilon \ll 1$, $\Theta$ is the `slow' phase,  $X_1 = \epsilon x_1, X_2= \epsilon x_2$ are the `slow' spatial variables and the local wavevector  $k = \nabla \theta = (\Theta_{X_1},\Theta_{X_2})$. Substituting this ansatz into a microscopic equations describing stripe patterns, we can derive an order parameter equation for the behavior of the phase $\Theta$, using the method of multiple scales. This was originally done by Cross and Newell \cite{cross1984convection}. 

The translation invariance $\theta(x_1,x_2) \to \theta(x_1,x_2) + \theta_0$ implies that the linearization of the dynamics in~\eqref{eq:sh} about a modulated stripe pattern has a non-trivial kernel, and the corresponding solvability condition yield, at lowest order, the {\em (unregularized) Cross-Newell equations} \cite{cross1984convection}, a gradient flow, that describes the macroscopic dynamics for the wave-vector field $k(x,t)$. These equations support shock formation and need to be regularized by higher order effects in the small parameter $\epsilon$ \cite{newell1996defects,ercolani2000geometry}.  

An alternative to employing the Fredholm alternative/solvability is to directly compute an effective energy $\mathcal{E}[k(x,t)]$ by averaging the energy~\eqref{eq:sh_energy} over all the microstates that are consistent with a given macroscopic field $k(x,t)$ \cite{newell2017elastic}. This is equivalent to averaging over the phase shift $\theta_0 \in [0,2 \pi]$, and yields the energy functional
\begin{equation}
\mathcal{E}_{RCN} = \int_{\Omega} \left((\nabla \cdot k)^2 + W(k) \right) dx
\label{eq:rcn}
\end{equation}
where $dx = dx_1 dx_2$ is two dimensional Lebesgue measure and $W$ is a nonconvex ``well potential" in $k$.  This \emph{regularized Cross-Newell}
(RCN) energy consists of two parts: a non-convex
functional of the gradient (the CN part) plus a convex regularizer whose energy density if  given by a quadratic 
function of the Hessian matrix $\nabla k = \nabla \nabla \theta$.

As discussed above, in order to describe stripe patterns, the rotated wavevector $k^{\perp}$ should be tangent to the leaves of a {\em measured foliation} \cite{poenaru1981some}. Consequently, a necessary condition for $k$ to describe a smooth stripe pattern is that $\nabla \times k = 0$. With the substitution $k = \nabla \theta$ (equivalent to the constraint $\nabla \times k = 0$), and the approximation $W(k) = (1-|k|^2)^2$, the RCN energy turns out to be closely related to the {\em Aviles-Giga} functional \cite{aviles1987mathematical}. 

Indeed, since $\int \left[(\Delta \theta)^2 - (\nabla \nabla \theta)^2\right] dx$ is a null Lagrangian, we get the same variational equations from the RCN energy with $W(k) = (1-|k|^2)^2$ and the Aviles-Giga energy functional, defined by
\begin{equation} \label{eq:gl}
{\mathcal{E}}_{AG} = \epsilon \int_\Omega \left(\nabla^2_{X}
 \Theta \right)^2 dX
      + \frac{1}{\epsilon} \int_\Omega (1 - |\nabla_{X} \Theta|^2)^2\, dX\ ,
\end{equation}
as expressed in terms of \emph{slow} variables stemming from
the modulational ansatz mentioned above: $X = (X_1,X_2) =
\left(\epsilon x_1, \epsilon x_2\right); \Theta = \epsilon \theta$. The parameter $\epsilon$ is the ratio between the wavelength of the pattern and the length scale at which the patterns deviate from being straight, parallel rolls. 

The energy functional~\eqref{eq:gl} only depends on the gradient and higher derivatives of $\Theta$, and is consequently invariant under the continuous symmetry $\Theta \rightarrow \Theta + \delta$, in contrast to the discrete  symmetry $\psi \rightarrow -\psi$ of the microscopic model  which corresponds to the discrete symmetries $\Theta \rightarrow \Theta + 2 k \epsilon \pi, k \in \mathbb{Z}$ and $\Theta \rightarrow -\Theta$ that follow from Eq.~(\ref{eq:F-symm}).

The variational problem associated with the the Aviles-Giga energy (\ref{eq:gl})
also arises in other physical contexts (unrelated to pattern
formation) \cite{aviles1987mathematical}. There is a considerable body of work on the Aviles-Giga variational problem, including ansatz-free lower bounds \cite{Jin2000Singular}, compactness in BV of the set $\{\Theta | \mathcal{E}_{AG}(\Theta) \leq C \}$ \cite{desimone2001compactness,ambrosio1999line} and sharp upper bounds \cite{conti2007sharp,poliakovsky2007upper}. Taken together, these results imply that (i) the energy in (\ref{eq:gl}) is scaled ``correctly" in the sense that the energy of the dominant singularities is $O(1)$ as $\epsilon \rightarrow 0$, (ii) The space of BV functions is (conjecturally) the right setting for looking at limits of families of minimizers as $\epsilon \rightarrow 0$, and (iii) The typical singularities for the limiting configurations  are 1-dimensional corresponding to a jump in the gradient.

Another natural energy functional for stripe patterns is given by the correspondence with smectic liquid crystals. In this correspondence, the director field $n$ is related to the phase field $\theta$ by $n = \pm \frac{\nabla \theta}{|\nabla \theta|}$. $n$ is thus well defined on the set $|\nabla \theta| \neq 0$. The `bending energy' is given by an appropriate norm of the curvature of the phase contours $\kappa = \nabla \cdot n$. In particular, we can replace the Hessian term in~\eqref{eq:gl} by the Oseen-Frank energy
\begin{equation}
    \label{eq:oseen-frank}
    \mathcal{E}_{OF} = \int \left|\nabla \cdot \left(\frac{\nabla \theta}{|\nabla \theta|} \right)\right|^2 dx
\end{equation}
In a region where $|\nabla \theta| =1$, it follows that $\kappa = \Delta \theta$ so the Oseen-Frank energy is identical to the Regularized Cross-Newell energy. These energies are however not equivalent in the vicinity of defects. For instance, the phase function $\theta = k_0^2 (x_1^2+x_2^2)$, corresponding to a target pattern has finite RCN and AG energy. However, the corresponding director field $n$ has $\nabla \cdot n = (x_1^2+x_2^2)^{-\frac{1}{2}}$ is not in $L^2_{\mathrm{loc}}$ on any neighborhood of the origin. Interpreted `strictly', the Oseen-Frank energy would seem to preclude the appearance of target patterns. For this reason, we will henceforth work with the RCN energy, and we expect that these results also inform the expectations for more realistic, physically motivated energy functionals since $|\nabla \theta| \approx 1$ over `large' sets in practice.

\subsection{Non-orientability and SBV phase fields} 

We now turn to the focus of this paper. We will only consider a fixed pattern wavelength and fixed domain $\Omega$, rather than a sequence of problems corresponding to limit $\epsilon \to 0$. Our primary interest is to develop a variational description that allows for disclinations, which are indeed seen in Rayleigh-B\'{e}nard convection far from threshold. As we discuss above, defects in natural patterns can be viewed as points with a non-trivial monodromy for $k=\nabla \theta$, that require a (locally) two-sheeted domain for a consistent and jump-free definition of $\theta$ and $k$ \cite{newell1996defects,ercolani2000geometry}. 

These point defects typically occur as convex-concave pairs corresponding to $+\frac{1}{2}$ and $-\frac{1}{2}$ degree singularities. Such defects are common in liquid crystals, among other physical settings, where the non-orientability is a consequence of the head-tail symmetry in the molecules. Mathematical models that describe the physical defects therefore have to contend with issues related to  orientability  \cite{ball2011orientability,bauman2012analysis,bedford2016function}. We will adapt some of these ideas for describing strip patterns in what follows.

 \begin{remark} \label{rem:phases}
Since the phase function $\theta$ is Lipschitz, Rademacher's theorem implies that $k=\nabla \theta$ is defined and equals the (density of the) distributional derivative $D\theta$ a.e. Further, in the ``nice" setting considered above $k$ is smooth on the strips $S_i$ and potentially has jumps on the rectifiable set $\Gamma_D$. We can relax the smoothness requirement to $k \in SBV^p(\Omega,\mathbb{R}^2)$, i.e. by requiring that the distributional gradient $Dk$ be a (symmetric) matrix-valued measure, that can be decomposed, $Dk = \nabla \nabla \theta + Jk$, into mutually singular measures, where $\nabla \nabla \theta \in L^p$ gives the part that is absolutely continuous with respect to 2-dimensional Lebesgue measure $\mathcal{L}^2$ and $Jk = \llbracket k\rrbracket \otimes \nu \mathcal{H}^1 \mres \Delta_k$ is concentrated on $\Delta_k \subseteq \Gamma_D$, the jump set of $k$. The index $p \geq 1$ ($p=2$ for us) encodes the control from having a finite bending energy for the pattern. 
\end{remark}

A natural approach to generalize the energy functionals defined for $W^{2,2}$ (i.e. defect-free) patterns, for instance the energies defined in Eqs.~\eqref{eq:rcn}, \eqref{eq:gl}~and~\eqref{eq:oseen-frank}, to SBV functionals that allow for the possibility of defects consists of two steps: (i) replacing the bending energy, initially given as an integral over the entire domain $\Omega$, by an integral over the complement of the jump set $\Omega \setminus \Delta_k$, and (ii) constraining the admissible phase fields so that the jump set $\Delta_k \subseteq \Gamma_D$, the set on which $\sin \theta = 0$. 

As an illustration, the generalization of the variational problem for the Cross-Newell energy is
\begin{equation}
\label{rcn_sbv}
    \theta^* = \arg \min_{\theta \in \mathcal{A}} \left[\int_{\Omega \setminus \Delta_k} (\Delta \theta)^2 d\mathcal{L}^2 + \int_\Omega (|\nabla \theta|^2 -1 )^2 d\mathcal{L}^2 \right],
\end{equation}
where $\Delta \theta$ is the (distributional) Laplacian of $\theta$, a $L^2_{\mathrm{loc}}$ function away from the defect set $\Delta_k$, and 
the  of admissible class  phase functions, $\mathcal{A}$, is defined by
\begin{equation}
\label{admissible_class}
     \mathcal{A} = \{\theta \in W_0^{1,\infty}, k:=D\theta \in SBV^2(\Omega,\mathbb{R}^2), \, \Delta_k \subseteq \Gamma_D\},
\end{equation}
where $\Delta_k$ is the jump set of $k$, and $\Gamma_D$ is the set on which $\sin \theta = 0$.

The first condition defining $\mathcal{A}$ is the requirement that $\theta$ be Lipschitz and satisfy a Dirichlet boundary condition on $\partial \Omega$. To elucidate the import of the second condition let us recall the definition of the space $SBV^2(\Omega, \mathbb{R}^2)$: 
\begin{enumerate}
    \item $k := D\theta \in SBV(\Omega,\mathbb{R}^2)$ if (the density) $k \in L^1_{\mathrm{loc}}(\Omega,\mathbb{R}^2)$ and the distributional derivative $Dk = DD\theta$ is a (symmetric matrix-valued) finite Radon measure.
    \item The measure $Dk := D^{ac} k + D^s k$ splits into two mutually singular pieces. $D^{ac} k$ is absolutely continuous with respect to Lebesgue measure $\mathcal{L}^2$.  $D^s k$ is singular with respect to Lebesgue measure and is purely a jump measure, i.e it has no Cantor part, and it is supported on a set with finite 1-dimensional Hausdorff measure. Specifically, $D^s k = \llbracket k \rrbracket \otimes \nu \mathcal{H}^1 \mres \Delta_k$, where the jump set $\Delta_k$ satisfies $\mathcal{H}^1(\Delta_k) < \infty \Rightarrow \mathcal{L}^2(\Delta_k) = 0$. $\nu$ is an $\mathcal{H}^1$ measurable unit normal to $\Delta_k$ and $\llbracket k \rrbracket = k_+ - k_-$ is the difference of the one-sided limits of $k$ on the jump set $\Delta_k$, where the $\pm$ sides are defined by the choice of the normal $\nu$.
    \item $k \in SBV^2(\Omega,\mathbb{R}^2)$ if $k \in SBV(\Omega,\mathbb{R}^2)$ and the density for the absolutely continuous part $D^{ac}k$ (with respect to $\mathcal{L}^2$), which will be denote by $\nabla k$ or $\nabla \nabla \theta$, is in $L^2(\Omega,M_{2\times 2}^{sym})$. 
\end{enumerate}
The final condition $\Delta_k \subseteq \Gamma_D$, restricts the (potential) jumps in $k$ to the set $\Gamma_D$, the zero set of $\sin \theta$. 
Additionally, depending on the specifics of the problem of interest, the phase function $\theta$ might be subject to additional boundary conditions. In this work, we will impose the Dirichlet boundary condition $\theta = 0$ on $\partial \Omega$. It is also sometimes useful to impose an additional Neumann condition $\partial_n \theta = -1$, where $n$ is the outward directed normal to $\partial \Omega$ \cite{Jin2000Singular,Ercolani2018phase}, although we will not do this here.

The appropriateness of considering this class $\mathcal{A}$ of admissible phase fields can also be motivated through a consideration of the symmetries of the microscopic (Swift-Hohenberg for example) model and those of the `defect-free' Cross-Newell energy~\eqref{eq:rcn}. Since the Cross-Newell energy functional~\eqref{eq:rcn} has a larger symmetry group $\theta \rightarrow \theta + \delta, \theta \rightarrow -\theta$, than the symmetries of the microscopic model $\theta \rightarrow \pm \theta + 2 j  \pi, j  \in \mathbb{Z}$, we need to restrict the class of admissible functions to ensure that the model has the right ``physical" symmetries \cite{ercolani2009variational}. This is achieved by the restriction $\Delta_k \subseteq \Gamma_D$ in the definition of $\mathcal{A}$. Indeed, $\Delta_k \subseteq \Gamma_D$ is invariant under the transformations $\theta \to \pm \theta + 2 j  \pi$, but not under arbitrary phase shifts $\theta \to \theta  + \delta$.

These considerations are also relevant to the question of defining appropriate energy functionals for layered structures in the context of liquid crystals. While some models  allow for disclinations at arbitrary phase values, i.e. keep the symmetry of arbitrary phase shifts $\theta \to \theta + \delta$ \cite{aviles1987mathematical,degennes1995book}, more detailed models explicitly account for the roles of the layers and half-layers, and the interaction of defects with the layer structure \cite{Pevnyi2014Modeling,Xia2021Structural}. It would be interesting to further explore this connection with energy functionals for pattern formation \cite{newell2017elastic,Zhang2021Computing}.

\subsection{Relaxation}

We define a relaxation of RCN energy functional by 
\begin{equation}
    \label{eq:erelax}
    \mathcal{E}_{SBV} = \int (\Delta \theta)^2 + (|\nabla \theta|^2 -1)^2 d\mathcal{L}^2 + \sigma \int |\sin(\theta)| \, |\llbracket \nabla \theta \rrbracket| \, d(\mathcal{H}^1 \mres \Delta_k),
\end{equation}
where $\nabla \theta \in SBV$, $J$ is the jump set for $\nabla \theta$, $|\llbracket \nabla \theta \rrbracket| = \llbracket \nabla \theta \rrbracket \cdot \nu$ and  $\sigma > 0$ is an $O(1)$ constant. This can be viewed as replacing the constraint $\Delta_k \subseteq \Gamma_D$ by a penalty $\sigma \int |\sin(\theta)| \, |\llbracket \nabla \theta \rrbracket| \, d(\mathcal{H}^1 \mres \Delta_k)$, yielding an unconstrained variational problem.

$\mathcal{E}_{SBV}$ is a relaxation of $\mathcal{E}_{RCN}$ since $\mathcal{E}_{SBV} \leq \mathcal{E}_{RCN}$ for all Lipschitz functions $\theta$ such that $\nabla \theta \in SBV^2$. If $D^2 \theta \in L^2$, the two functionals are identical since $D^2 \theta = \nabla^2 \theta$ and the jump set of $\nabla \theta$ is empty. Conversely, if $\nabla \theta \in SBV^2(\Omega)$ but is not in $H^1$, the RCN energy equals $+\infty$ so that $\mathcal{E}_{SBV} \leq \mathcal{E}_{RCN}$.

Note that the relaxed energy can be finite even if $\theta \notin \mathcal{A}$, since $\mathcal{E}_{SBV} < \infty$ does not imply that $\int |\sin(\theta)| |\llbracket \nabla \theta \rrbracket| d\mathcal{H}^1 = 0$, or equivalently that $\Delta_k \subseteq \Gamma_D$. Consequently, the admissible set for $\mathcal{E}_{SBV}$ is strictly larger than $\mathcal{A}$. By the general `philosophy' of relaxation, in considering the relaxed functional, we are enlarging the admissible set thereby improving the likelihood of finding a minimizer. If the minimizer for the relaxed problem \eqref{eq:erelax} is also in $\mathcal{A}$, i.e. is admissible for the original problem~\eqref{rcn_sbv}, then we can conclude that we have also found a minimizer for the original problem.

If $\theta_1 \sim \theta_2$ are two equivalent phase functions, then $\mathcal{E}_{SBV}[\theta_1] = \mathcal{E}_{SBV}[\theta_2]$. In other words, the functional $\mathcal{E}_{SBV}$ respects the principle that $h = \cos \theta$ is physical and that gauge transformations, as in Thm.~\ref{thm:orient}, should not affect the energetics. For the purposes of computation and mathematical analysis, however, there are good reasons to `break the degeneracy' and pick a unique representative phase. One possible choice is the canonical representative $\tilde{\theta} = \arccos(\cos(\theta))$ as we discussed above. The range of $\tilde{\theta}$ is restricted to be in $[0,\pi]$, and $\tilde{\theta}$ has the maximal possible jump set, $\Delta_k = \Gamma_D$ \cite{Machon2019Aspects}. While this choice is unique, it has the unfortunate feature that, as the ratio of the wavelength to the size of the domain, $\epsilon \to 0$, the phase function will display bounded oscillations on increasingly finer scales, so that the weak limit will wash out all the information about the phase function. 

An alternative is to pick the oriented representative $\bar{\theta}$ that has the `smallest' allowed jump set $\Delta_k$ as we discuss in \S\ref{sec:orient}. We can break the `gauge degeneracy' of the relaxed energy function $\mathcal{E}_{SBV}$, leading to following the variational problem over the admissible set $\theta \in W_0^{1,\infty}, D\theta \in SBV^2$:
\begin{equation}
\label{eq:RSBV}
    \theta^\delta = \arg \min_\theta \left(\mathcal{E}_{SBV}[\theta] + \delta \mathcal{H}^1 (\Delta_k) \right),
\end{equation}
where $\delta$ is `small' and, ideally, we would want to take the limit $\delta \to 0^+$. This procedure is entirely analogous to the addition of the bending energy with a small coefficient $\epsilon$ to regularize unphysical behavior in the `bare' Cross-Newell energy functional \cite{newell1996defects, newell2017elastic} to obtain the regularized Cross-Newell energy~\eqref{eq:rcn}. 

As we will argue below, because of a regularizing effect akin to {\em numerical viscosity}, our method to solve the variational problem for the functional in~\eqref{eq:erelax} actually ends up solving the variational problem for the regularized functional~\eqref{eq:RSBV} with a small $\delta > 0$. This breaks the gauge degeneracy, and gives the sought for oriented representative $\bar{\theta}$.

\section{Minimizing movements for the SBV Cross-Newell energy}
\label{sec:bregman}

Our goal is to numerically discretize the energy in~\eqref{eq:erelax} and then minimize the resulting discretization. A significant hurdle is  discretizing the splitting 
\begin{equation}
  Div(D \theta)  = \Delta \theta \, d\mathcal{L}^2 + \llbracket \nabla \theta \rrbracket\cdot \nu \, d(\mathcal{H}^1 \mres \Delta_k),
  \label{eq:splitting}
\end{equation} 
where $D$ is the distributional gradient and $Div$ is the distributional divergence.

One approach to this problem, following the the ideas of Ambrosio and Tortorelli \cite{Ambrosio1992Approximation}, is to introduce an additional field $\phi^\epsilon$ and an additional length scale $\epsilon \ll 1$. The splitting into absolutely continuous and singular (w.r.t. Lebesgue measure) is effected through this additional field $\phi^\epsilon$ which has the following properties:
\begin{enumerate}
\item $\phi^\epsilon(x) \in [0,1]$ for all $x$ in the domain and for a fixed threshold $t \in (0,1)$, say $t = \frac{1}{2}$, we have that the measure  $\mathcal{L}^2(\{x : \phi^\epsilon(x) < t\})$ is $O(\epsilon)$.
\item $\phi^\epsilon$ is the minimizer for a Modica-Mortola type energy functional with length scale $\epsilon >0$, and is an ``approximate'' indicator function for the support of the absolutely continuous part in the splitting \cite{Ambrosio1992Approximation}. Equivalently, for a given $t \in (0,1)$, the sublevel set $\phi^\epsilon < t$  is an ``$\epsilon$-level approximation" of the free discontinuity $\Delta_k$. In other words, corresponding to the splitting $\lambda = \lambda_{ac} + \lambda \mres \Delta_k$, we have, at the $\epsilon$-level, and  for any given bounded continuous function $\psi$,
$$
    \int \psi d\lambda_{ac}  \approx \int \psi \phi^\epsilon d\lambda^\epsilon, \qquad
    \int \psi d(\lambda \mres \Delta_k)  \approx \int \psi (1-\phi^\epsilon) d \lambda^\epsilon,
$$
where $\lambda^\epsilon$ is absolutely continuous with respect to Lebesgue measure, but becomes increasingly singular as $\epsilon \to 0$ and converges in measure (i.e. weak-$*$) to $\lambda$ \cite{Ambrosio1992Approximation}.
\end{enumerate}

We use an alternate approach, based on treating the energy as a functional of the ``decoupled" quantities $\theta$,  $\rho = \Delta \theta$, and $d\mu = \llbracket \nabla \theta \rrbracket\cdot \nu d(\mathcal{H}^1 \mres \Delta_k)$, and minimizing subject to the {\em linear} constraint
\begin{equation}
\label{eq:constraint}
   -\int_{\Omega} \nabla \varphi \cdot \nabla \theta d\mathcal{L}^2  = \int \rho \varphi d\mathcal{L}^2 + \int \varphi \,d \mu,
\end{equation}
for all $\varphi \in C_c^\infty(\Omega)$. Eq.~\eqref{eq:constraint} is the weak form of~\eqref{eq:splitting}. In obtaining this result, we have used the fact that $\theta$ is Lipschitz so that $D \theta = \nabla \theta$. The energy in \eqref{eq:erelax} can be recast as 
\begin{equation}
    \label{eq:rlx2}
    \mathcal{E}_{SBV} = \int \left[ \rho^2 + (|\nabla \theta|^2 -1)^2\right] d\mathcal{L}^2 + \sigma \int |\sin(\theta)| \, d|\mu|,
\end{equation}
and is therefore the sum of an $L^2$ norm on $\rho$, a ``TV like" term in $d\mu$, and a lower-order (relative to $\rho = \Delta \theta$) nonconvex term in $\nabla \theta$. Following the approach in \cite{Jaramillo2021Modified}, we can treat this variational problem by exploiting convexity splitting along with the  split Bregman method. 

\subsection{The split Bregman IMEX algorithm}\label{s:SBalgorithm}

The Bregman algorithm \cite{Bregman1966Relaxation} is an iterative approach to solve the variational problem 
\begin{equation}
    u^* = \lim_{\lambda \to + \infty} \arg\min_{u \in X} (\lambda^{-1} J(u) +  H(u))
    \label{relaxation_variational}
\end{equation}
where $X$ is a Banach space with a norm we will henceforth denote by $\|\cdot \|$, and  $J: X \to [0,\infty]$ and $H:X \to [0,\infty]$ are convex. $H$ need not have a unique minimizer and, in this context, $J$ is a {\em regularizer}. Intuition suggests that $u^*$ minimizes $J$ {\em among the mimimizers of $H$}. For our applications we will further assume that  
\begin{enumerate}
\item $\{u \in X \, | \, H(u) = 0, J(u) < \infty \} \neq \emptyset$, 
\item $H$ is Fr\'{e}chet-differentiable so that the subgradient  $\partial H = \nabla H$, the Fr\'{e}chet derivative, and 
\item the sublevel sets  $J[u] \leq C$ are precompact in $X$. 
\end{enumerate}
The third assumption automatically holds for coercive functionals on the finite dimensional space $\mathbb{R}^n$. The difficulty with a direct approach to this regularized variational problem, say using gradient descent, is that as $\lambda$ gets large, the resulting system of equations can become poorly conditioned. 

The {\em Bregman distance} or {\em Bregman divergence} (for the functional $J$) is defined by
$$
D^p_J (u,v) = J(u)-J(v)-\langle p,u-v\rangle \quad \mbox{where } u,v \in X, p \in \partial J(v)
$$
and $\partial J(v)$ denotes the subdifferential of $J$ at the point $v$. Since $J$ is convex, $D^p_J(u,v) \geq 0$ and if $J$ is strictly convex, we have $D^p_J(u,v)=0$ only if $u=v$. The Bregman distance is convex in the argument $u$, although it is not a metric on $X$ since it is not necessarily symmetric in $u$ and $v$. The Bregman iteration is a primal-dual algorithm that starts with a pair $(u_0,p_0)$ where $p_0 \in \partial J( u_0)$ and updates $(u_k,p_k)$ by the rule \cite{Bregman1966Relaxation}
\begin{equation}
\label{bregman_iteration}
    u_{k+1} = \arg\min_{u \in X} \left[D^{p_k}_J(u,u_k) + \lambda H(u)\right], \quad p_{k+1} = p_k - \lambda \nabla H(u_k).
\end{equation}
It follows from a direct calculation that $p_{k+1}$ as defined by the iteration is indeed in the subgradient of $J$ at $u_{k+1}$, and further we have the estimate $H(u_k) \leq \frac{J(\tilde{u})}{\lambda k}$ where $\tilde{u} \in X$ is any minimizer of $H$ satisfying $H(\tilde{u}) = 0$ and $J(\tilde{u}) < \infty$. We refer the reader to the original article by Bregman \cite{Bregman1966Relaxation}, and to more recent work that builds on these ideas \cite{osher2005,yin2008bregman}, for the proofs of these claims. The compactness assumption, along with the direct method, shows that, for a subsequence, we have $u_{k_n} \to u^*$, thereby solving the regularized variational problem~\eqref{relaxation_variational}.

An important feature of this method is that, if $J$ is strictly convex, the iteration in~\eqref{bregman_iteration} converges for any choice of $\lambda > 0$ and we can thus choose $\lambda$ to improve the condition number of the minimization over $u$ \cite{osher2005} rather than necessarily sending $\lambda \to +\infty$ as in a direct approach to~\eqref{relaxation_variational}.

This framework can naturally handle convex optimization subject to linear constraints \cite{Bregman1966Relaxation,osher2005}. For our purposes, it suffices to consider the finite dimensional (i.e. discretized) setting. Given a convex functional $E:\mathbb{R}^n \to \mathbb{R}$ and a linear constraint $Au=b \in \mathbb{R}^m$, it follows from our discussion above that the Bregman iteration with the choices $X = \mathbb{R}^n, J(u)=E(u), H(u) = \frac{1}{2}\|A u - b\|^2$ will solve the constrained optimization problem $u = \arg \min E(u)$ subject to $Au = b$. In this case, the primal-dual Bregman iteration~\eqref{bregman_iteration} can be elegantly recast \cite{yin2008bregman} as 
\begin{equation}
\label{error_correction}
    u_{k+1} = \arg \min_u \left( E(u) + \frac{\lambda}{2} \|A u - b_k\|^2 \right), \quad b_{k+1} = b_k+b - A u_k.
\end{equation}
The second step in the iteration defines $b_{k+1}$ and is also called the {\em Bregman update} \cite{yin2008bregman,goldstein2009}. It amounts to adding  the residual $b-A u_k$ back to $b_k$. This feature, of ``adding back the noise" is key to the {\em error forgetting property}  of the Bregman method as adapted to convex optimization with linear constraints \cite{yin2013error}. In particular, this yields good performance for the algorithm even if the first step in the iteration, namely optimizing the augmented functional over $u$, is only carried out approximately \cite{osher2005,yin2013error}. For any convex function $E: \mathbb{R}^n \to \mathbb{R}$, the iteration~\eqref{error_correction} is provably convergent \cite[Thm 2.2]{osher2005}.

\begin{remark}
There are two further impediments to applying this approach to the functional $\mathcal{E}_{SBV}$ in~\eqref{eq:erelax}. They are
\begin{enumerate}
    \item The functional depends on various orders of derivatives of the phase field $\theta$ and further Radon-Nikodym splittings of these derivatives. This splitting is not straightforward to discretize. Additionally, the first step in the iteration~\eqref{error_correction}, namely minimizing the functional over $\theta$ will yield a very complicated 4-th order equation.
    \item The functional is {\em not} convex in $\theta$.
\end{enumerate}
\end{remark}

The first issue is addressed by the {\em split Bregman algorithm} \cite{goldstein2009} that treats the various orders of derivatives as independent variables that are related by linear (differential) constraints. For example, to minimize the functional $W[u_x] + V[u]$ where $W,V$ are convex functionals given by integrating a local energy density, the ideas is to instead consider the equivalent formulation,
\begin{equation}\label{e:constrain}
 \min_{u,d}  W[d] + \mathcal{V}[u] \quad \mbox{subject to} \quad u_x = d.
 \end{equation}
 We can solve this using the iteration~\eqref{error_correction} which gives
\begin{align*}
u^{k+1} & = \arg\min_{u,d}  \mathcal{V}[u] + \frac{\lambda}{2} \| d^k - u_x -b^k\|^2,\\
d^{k+1} & = \arg\min_{u,d} W[d] + \frac{\lambda}{2} \|d - u_x^{k+1}  -b^k\|^2,\\
b^{k+1} & =  b^k + (u_x^{k+1} - d^{k+1})
\end{align*}
The advantages of this method, as well as its application to problems in image processing, are discussed in \cite{goldstein2009,yin2013error}.
 
To address the second concern, the non-convexity of $\mathcal{E}_{SBV}$, we combine the ideas behind minimizing movements \cite{Ambrosio1995Minimizing}, IMEX (convexity splitting) schemes \cite{Ascher1995ImplicitexplicitMF,glasner2016improving} and split Bregman methods, in the spirit of the work in \cite{Jaramillo2021Modified}. The notion of minimizing movements, originally formulated by De Giorgi, generalizes the notion of gradient flows to energy functionals that are not necessarily differentiable, although that still possess sufficient structure to imply the existence of minimizers. In particular, given a coercive, lower semicontinuous functional $\mathcal{E}$ on a Banach space $X$, and a time-step $\tau$, evolution by minimizing movements is given by the sequence 
\begin{equation}
    u_{k+1} = \arg \min_{u \in X} \left[\frac{1}{2 \tau} \|u-u_k\|^2 + \mathcal{E}[u]\right].
    \label{min_moves}
\end{equation}
If $\mathcal{E}$ is Fr\'{e}chet differentiable, this scheme can be thought of as an implicit (backward Euler) time-stepping for the gradient flow $u_t = - \frac{\delta}{\delta u} \mathcal{E}$. However, in our setting, $\mathcal{E}_{SBV}$ is not convex so we have to adapt the method of minimizing movements. Consider a (non-unique!) splitting $\mathcal{E} = \mathcal{E}^+ + \mathcal{E}^-$ of a functional $\mathcal{E}$ into a convex and a concave part. For a differentiable functional $\mathcal{E}$, the weak formulation of the gradient flow is 
\[ \langle u_t,w \rangle  = - \left \langle \frac{\delta \mathcal{E}^+}{\delta u}[u],w \right \rangle -\left\langle \frac{\delta \mathcal{E}^-}{\delta u}[u],w \right \rangle. \]
In the convexity splitting scheme, the contribution of the nonconvex part $\mathcal{E}^-$ is treated explicitly in the time stepping, i.e. it is evaluated at a previous time step $t = t_k$ and treated as a forcing term. For a given time step $\tau$, the IMEX (implicit-explicit) algorithm for discretizing the gradient flow \cite{Ascher1995ImplicitexplicitMF} is given by,
\[ \left\langle \frac{u_{k+1}-u_k}{\tau}, w \right\rangle  =- \left \langle  \frac{\delta \mathcal{E}^+}{\delta u}[u_{k+1}],w \right \rangle - \left \langle \frac{\delta \mathcal{E}^-}{\delta u}[u_k],w \right \rangle.\]
This equation is formally an Euler Lagrange equation, and motivates a generalization of the minimizing movements defined in~\eqref{min_moves} through
\begin{equation}
    \label{e:Rayleigh}
u_{k+1} = \arg \min_{u \in X} \left[\frac{1}{2 \tau} \|u-u_k\|^2 + \mathcal{E}^+[u] + \left \langle \frac{\delta \mathcal{E}^-}{\delta u}[u_k],u-u_k \right \rangle + \mathcal{E}^-[u_k]\right].
\end{equation}
where we have assumed that $\mathcal{E}^-$ is Fr\'{e}chet-differentiable. Note that the last two terms correspond to  the linearization of $\mathcal{E}^-$ at $u_k$, so that the `effective' energy functional in~\eqref{e:Rayleigh} is convex. We can now apply the modified split Bregman algorithm described above to effect the minimization to find $u_{k+1}$. Owing to the presence of the term $\frac{1}{2 \tau} \|u-u_k\|^2$, the objective in ~\eqref{e:Rayleigh} is strictly convex and minimizers exist and are unique. Moreover, as shown in \cite{glasner2016improving}, for the case where $\mathcal{E}$ is differentiable, the sequence $\{u_k\}$  converges to a local minimum of $\mathcal{E}$ to within an error of $O(\tau)$ that is set by the size of the time-step.

\subsection{Finite element discretization}

We now discuss the discretization of the SBV relaxation of the Cross-Newell energy~\eqref{eq:erelax} and a numerical implementation of the method of minimizing movements along the lines of the discussion in \S~\ref{s:SBalgorithm}. Some of the following arguments are well known in the context of piecewise linear finite elements \cite{Brenner2008FEM} but we include them here to keep the discussion self contained. 

We start with a triangulation $\mathcal{T} = (V,E,F)$ of the closure $\bar{\Omega}$ of a given domain $\Omega$. Here $V,E$ and $F$ are respectively the vertices, edges and faces of the triangulation. We will denote the set of interior vertices by $V^0 \subseteq V$. Examples of  triangulations of an elliptical domain is shown in Fig.~\ref{fig:triangulate}.

\begin{figure}
    \centering
    \includegraphics[width=0.7\textwidth]{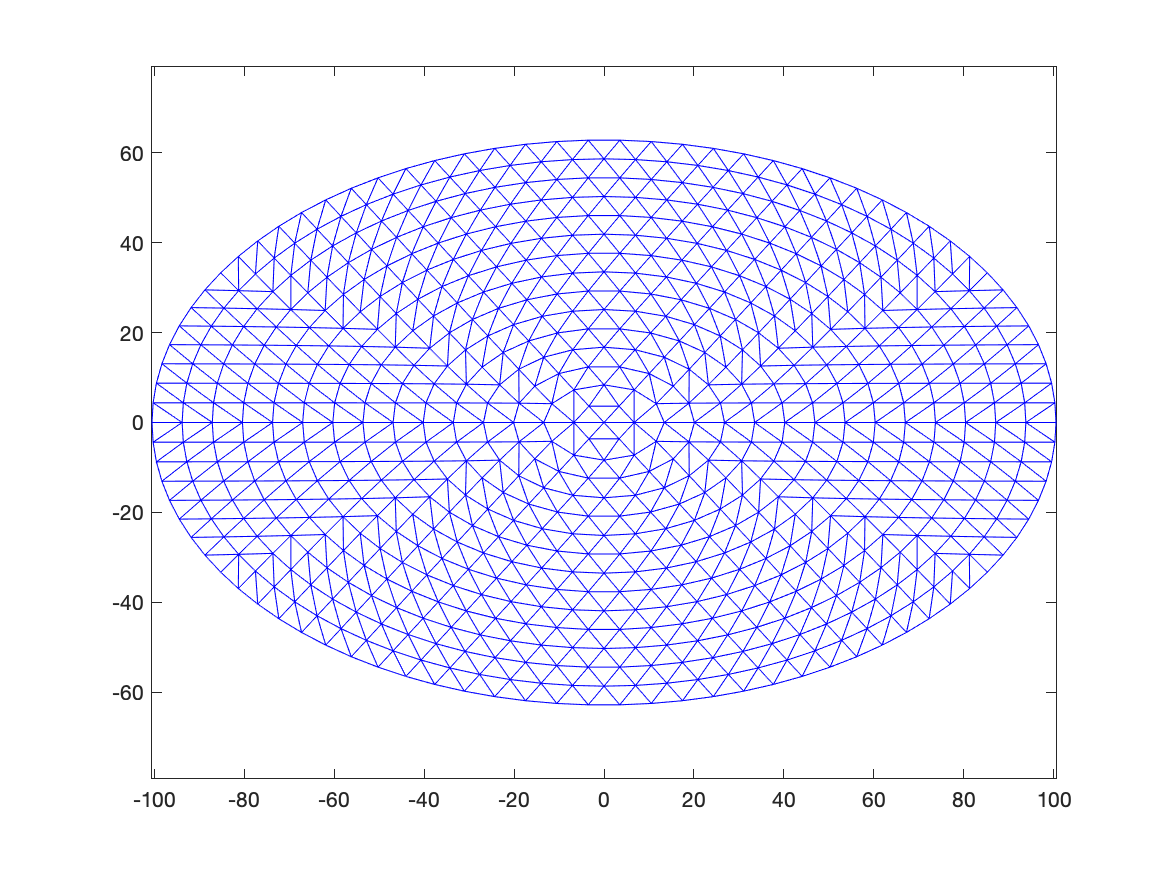} 
    \caption{A triangulation of the elliptical domain $\left(\frac{x_1}{32 \pi}\right)^2 + \left(\frac{x_2}{20 \pi}\right)^2 < 1$. For purposes of easy visualization, this mesh is 16 times coarser than the ones we use in our numerical simulations.}
    \label{fig:triangulate}
\end{figure}

After picking a triangulation, we will, if needed, slightly modify the domain $\Omega$ so that $\Omega$ is the union of triangles with straight edges. In particular, this replaces the original `smooth' curved portions of the boundary $\partial \Omega$ by piecewise linear curves.

Let $W$ denote the set of all continuous functions on $\Omega$ that vanish on $\partial \Omega$ and are piecewise linear on each face $F$ of the triangulation. This space is clearly finite dimensional, and determined by the values of the function on all the interior vertices. $f \in W$ implies that $f = \sum_i f_i \varphi_i$ where $f_i = f(v_i)$ for all vertices $v_i \in V^0$ and $\varphi_i$ is the piecewise linear `tent' function satisfying $\varphi_i(v_j) = \delta_{ij}$. More generally, the projection operator $f \mapsto \mathbf{F} := \{f(v_i)\}_{v_i \in V^0}$ gives a mapping from $C_0(\bar{\Omega})$ to $\mathbb{R}^n$, i.e. column vectors with $n$ entries. We will denote column vectors by the boldface latin alphabet $\mathbf{F}, \mathbf{G}, \ldots$. Here $n = |V^0|$ denotes the number of interior vertices in the triangulation $\mathcal{T}$. The sampling operator has an ``approximate inverse" given by the interpolation operator $\mathbf{F} = \{f_i\} \mapsto \sum f_i \varphi_i \in C_0(\bar{\Omega})$. 

A Borel measure on $\Omega$ can now be approximated as a linear mapping $W \to \mathbb{R}$. Appealing to the isomorphism $W \cong \mathbb{R}^n$, we deduce that Borel measures are represented as elements in the dual of $\mathbb{R}^n$, i.e. row vectors with $n$ entries, and the natural pairing is given by 
$$
\int_\Omega f d\mu \approx \sum_i \mu_i f_i = \mathbf{\mu}^T \mathbf{F}
$$

A function $g \in L^1_{\mathrm{loc}}$ corresponds to a Borel measure $\mu_g$ through $\mu_g(f) = \int f g d\mathcal{L}^2$ for all continuous functions $f \in C_0(\bar{\Omega})$. If $g$ is also in $C_0(\bar{\Omega})$, we have the quadrature rule 
$$
\int f g d\mathcal{L}^2 \approx \sum_{ij} g_i f_j \int_\Omega \varphi_i \varphi_j d \mathcal{L}^2 =   g_i M_{ij} f_j = \mathbf{G}^T \mathbf{M} \mathbf{F}, \quad M_{ij} = \int_\Omega \varphi_i \varphi_j d\mathcal{L}^2,
$$
where $\mathbf{M}$ is a symmetric `Mass' matrix and we have the correspondence $g d\mathcal{L}^2 \mapsto \mu_g =\mathbf{G}^T \mathbf{M}$.

Piecewise linear continuous functions on the triangulation $\mathcal{T}$ are Lipschitz, so there is a well defined gradient operator $\nabla$ on $W$. The pointwise gradient is piecewise constant on each face of the triangulation, so it can be represented by two $|F| \times |V^0|$ matrices, $\mathbf{D}_x$ and $\mathbf{D}_y$ that give, respectively, the $x$ and $y$ components of the gradient operator. We can represent the gradient $\nabla f$ by an element in $\mathbb{R}^{2 |F|}$ given by 
\begin{equation}
    \nabla \mathbf{F} = \begin{bmatrix} \mathbf{D}_x \\ \mathbf{D}_y \end{bmatrix} \mathbf{F} 
    \label{gradient_op}
\end{equation} 

Taking $\varphi = \varphi_i, i = 1,2, \ldots,n$ in turn, Eq.~\eqref{eq:constraint} yields the system of constraints
\begin{align}
     -\int (D \varphi_i) \cdot (D \theta) \, d \mathcal{L}^2 & = \int \rho \varphi_i \, d \mathcal{L}^2 + \int \varphi_i d\mu \nonumber \\
     \implies -\sum_j K_{ij} \theta_j & = \sum_j M_{ij} \rho_j + \mu_i,
\end{align}
 where $K_{ij} = \int (\nabla \varphi_i) \cdot (\nabla \varphi_j) d \mathcal{L}^2 $ is the {\em stiffness} matrix. A direct calculation yields
 $$
  \mathbf{K} = \mathbf{D}_x^T \mathbf{A} \mathbf{D}_x + \mathbf{D}_y^T \mathbf{A} \mathbf{D}_y
 $$
 where $\mathbf{A}$ denotes the $|F| \times |F|$ diagonal matrix whose entries are  given by the areas of the faces, $A_{ij} = \mathcal{L}^2(F_i) \delta_{ij}$. Here and henceforth $\delta_{ij}$ denotes the Kronecker delta.

The remaining consideration is to identify an appropriate convexity-splitting for the functional in~\eqref{eq:rlx2}, that rewrites $\mathcal{E}_{SBV}$ in terms of `decoupled' variables $\theta, \rho$ and $d\mu$. We have,
\begin{align*}
    \int_\Omega (|\nabla \theta|^2 -1)^2 \, d\mathcal{L}^2 & = \mathcal{L}^2(\Omega) + a \int |\nabla \theta|^2 d \mathcal{L}^2 - \int \left[(a+2)|\nabla \theta|^2 - |\nabla \theta|^4\right] d\mathcal{L}^2 
\end{align*}
so that the first term is a constant and can be dropped, the second term is convex and will be a part of $\mathcal{E}^+$, and the third term is concave if we have the uniform estimate $\|\nabla \theta\|_{\infty} < \frac{a+2}{6}$. We have $|\nabla \theta| \lesssim 1$ for stripe patterns, so it suffices to take $a = 6$ in our computations.
Also, the term $\sigma \int |\sin(\theta)| d|\mu| \approx \sum_i |\mu_i| |\sin(\theta_i)|$ is convex in $\mu_i$ (with fixed $\theta$) but is not necessarily convex in $\theta$. We will handle this term implicitly in the update for $\mu$ but treat it explicitly in the update for $\theta$.

With all of these considerations, we have the following split-Bregman algorithm for the evolution of stripe patterns with (potential) disclinations. We initialize by $\theta^0(x) = \mathrm{dist}(x,\partial \Omega), \rho^0 = 0, \mu^0 = 0, b^0 = 0$. We set $\lambda =0.5$ in our computations. In general, $\lambda$ can be chosen by trial and error to speed up the rate of convergence.
\begin{align}
    \rho^{k+1} & = \arg \min_{\rho \in \mathbb{R}^n} \left[\rho^T \mathbf{M} \rho + \frac{\lambda}{2} \|\mathbf{M} \rho + \mathbf{K} \theta^k + \mu^k + b^k\|^2 \right] \nonumber \\
    \mu^{k+1} & = \arg \min_{\mu \in \mathbb{R}^n} \left[\sigma \sum |\mu_i| |\sin(\theta_i^k)|  + \frac{\lambda}{2} \|\mathbf{M} \rho^{k+1} + \mathbf{K} \theta^k + \mu + b^k\|^2 \right]\nonumber \\
    \zeta^{k} & = \arg \min_{\zeta \in \mathbb{R}^n} \left[\frac{1}{2\tau} \|\zeta - \theta^{k}\|^2 + a \zeta^T \mathbf{K} \zeta - \zeta^T (2(a+2)\mathbf{K} - 4 \mathbf{L}^k) \theta^{k}\right. \nonumber \\ 
    & \quad \left. + \frac{\lambda}{2} \|\mathbf{M} \rho^{k+1} + \mathbf{K} \zeta + \mu^{k+1} + b^k\|^2  \right]\nonumber \\
    \theta^{k+1} & = \arg \min_{\theta \in \mathbb{R}^n} \left[\frac{1}{2\tau} \|\theta - \zeta^k\|^2 + \sigma \sum |\mu^{k+1}_i| |\sin(\zeta^k_i)| \right] \nonumber \\
    b^{k+1} & = b^k + \mathbf{M} \rho^{k+1} + \mathbf{K} \theta^{k+1} + \mu^{k+1}
    \label{eq:SB_IMEX}
\end{align}
where we are time-stepping $\theta$ using operator splitting as $\theta^k \to \zeta^k \to \theta^{k+1}$, and the last step is the Bregman update. The matrix $\mathbf{L}^k$, which gives the linearization of the concave part $\mathcal{E}^-$, is updated along with the other quantities, and is given by
\begin{align}
    \mathbf{L}^k = \mathbf{D}_x^T \mathbf{B}^k \mathbf{D}_x + \mathbf{D}_y^T \mathbf{B}^k \mathbf{D}_y, \nonumber \\
    B^k_{ij} = \mathcal{L}^2(F_i)\left[ (\mathbf{D}_x \theta^k)_i^2 + (\mathbf{D}_y \theta^k)_i^2 \right] \delta_{ij}
\end{align}

$\rho^{k+1}$ and $\zeta^k$ are defined by linear least squares, and the matrices $\mathbf{M}$ and $\mathbf{K}$ are sparse. Also, we don't need to solve them exactly on account of the error forgetting property of the Bregman iteration. We therefore use a few steps of the Conjugate gradient method to update $\rho$. For updating $\zeta$, we exploit the fact that the linear system is strongly diagonally dominant if the time step $\tau$ is small, and we can efficiently obtain an approximate solution with a few Gauss-Seidel steps. For the results presented below, we use five conjugate gradient iterations for updating $\rho$ and five Gauss-Seidel iterations for updating $\zeta$.

The update for $\mu_i$ can be done exactly and efficiently using a shrink operator as follows:
\begin{align}
    \nu & = -\mathbf{M} \rho^{k+1} - \mathbf{K} \theta^k - b^k  \nonumber \\
    \mu^{k+1}_i & = \mathrm{sign}(\nu_i) \max \left(|\nu_i| - \frac{\sigma}{\lambda} |\sin(\theta_i^k)|,0\right)  
     \label{mu_update}
\end{align}
As with $\mu$, the update for $\theta^{k+1}$ can also be done componentwise. A formal calculation indicates that, for each interior vertex $i \in V^0$, the relevant gradient flow is $ \theta_i(0) = \zeta^k_i,  \frac{d\theta_i}{d t} = - \frac{d}{d\theta} \sigma |\sin(\theta_i)| = - \sigma \, \mathrm{sign}(\sin(\theta_i)) \cos(\theta_i)$. This formal calculation, however, is only valid as long as $\sin(\theta) \neq 0$. 

The minima of $\sigma |\sin(\theta)|$, the potential driving the gradient flow, are at $\theta = j \pi, j \in \mathbb{Z}$. The rate at which $\theta(t)$ approaches these minima is non-vanishing, so it is possible to reach them in finite time, unlike the case of a gradient flow near a smooth minimum.  If $\theta_i(t) = j \pi $ for some $j \in \mathbb{Z}$ at any point in the interval $[0,\tau]$, $\sin(\theta)$ vanishes within a time step and the correct update is to set $\theta^{k+1}_i = \theta_i(\tau) = j \pi$. Using a forward Euler timestep for the formal gradient flow, along with the correction to prevent overshoot, gives an update rule for $\theta$ that has some similarities with~\eqref{mu_update}:
\begin{align}
    \beta_i & = \pi \times \left \lfloor \frac{\zeta_i^k}{\pi} + \frac{1}{2} \right \rfloor \nonumber \\
    \theta^{k+1}_i & = \zeta^k_i - \mathrm{sign}(\sin(2\zeta^k_i)) \min \left( \sigma \tau |\cos(\zeta^k_i)|,|\beta_i -\zeta^k_i|\right)  
     \label{theta_update}
\end{align}
where $\lfloor \cdot \rfloor$ denotes the floor function so that $ x \mapsto \lfloor x + \frac{1}{2} \rfloor$ maps a real number to the nearest integer (breaking ties by $j+\frac{1}{2} \mapsto j+1$ for $j \in \mathbb{Z}$).

\subsection{Numerical results} We illustrate the split Bregman IMEX approach to numerically solving the variational problem~\eqref{eq:erelax} using the example of an elliptical domain with semi-major axis $a = 32 \pi$ and semi-minor axis $b = 20 \pi$. We first iterate the system~\eqref{eq:SB_IMEX} with the restriction $\mu^k_i = 0$ at each step, i.e. we prescribe `by fiat' that there are no jumps in $k = \nabla \theta$. This is equivalent to minimizing $\mathcal{E}_{RCN}$ in~\eqref{eq:rcn} over the class of $W^{2,2}$  phase fields without any disclinations, or alternatively, to minimizing the Aviles-Giga functional~\eqref{eq:gl}. The results in \cite{Jin2000Singular} imply that, outside of a boundary layer with $O(1)$ width, the minimizer is well approximated by the $\bar{\theta}(x) = \mathrm{dist}(x,\partial \Omega)$, the distance to the boundary. Consequently, $|\nabla \theta| \approx 1$ and $\theta$ is a solution of the eikonal equation over the bulk of the domain, outside a narrow boundary layer where $|\nabla \theta|\approx 0$. These expectations are borne out by the results of our numerical method, as shown in the top row of Fig.~\ref{fig:AG} 

\begin{figure}
    \centering
    \includegraphics[height=1.8in]{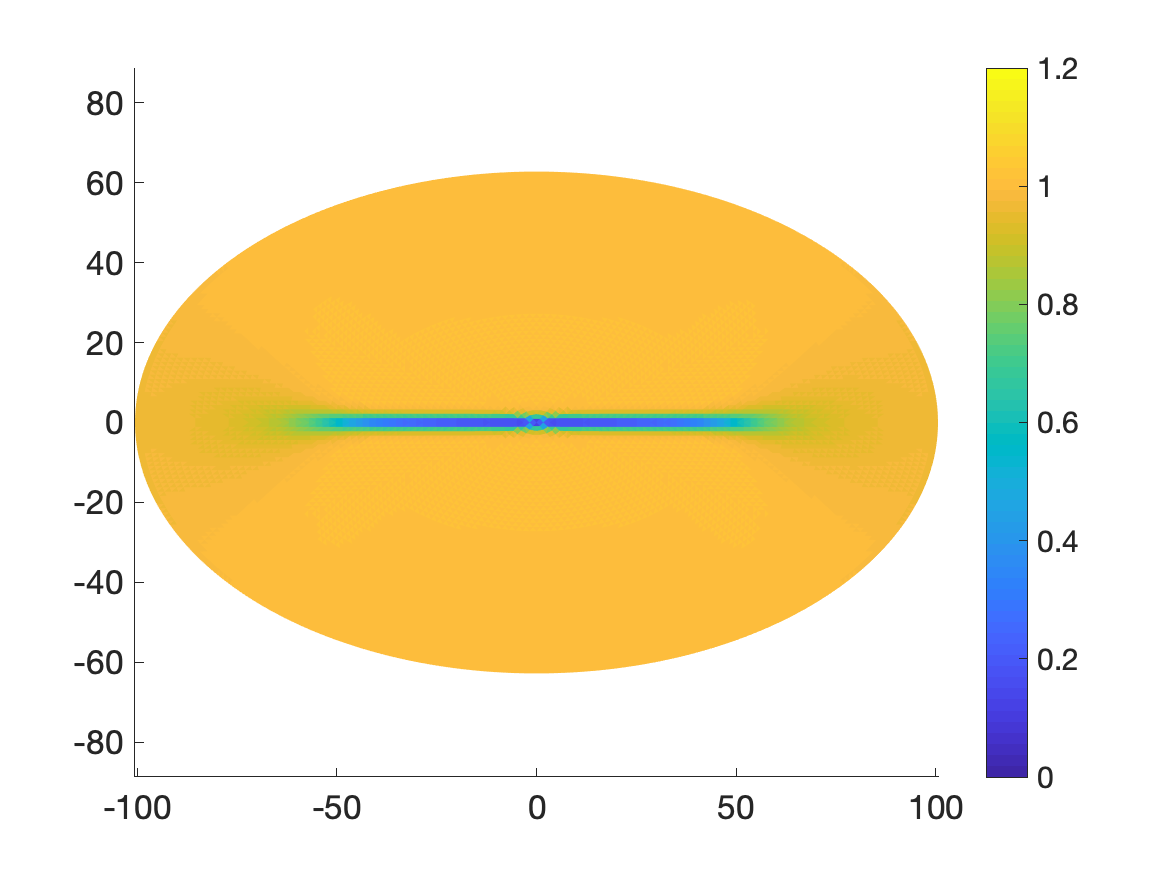} 
    \includegraphics[height=1.8in]{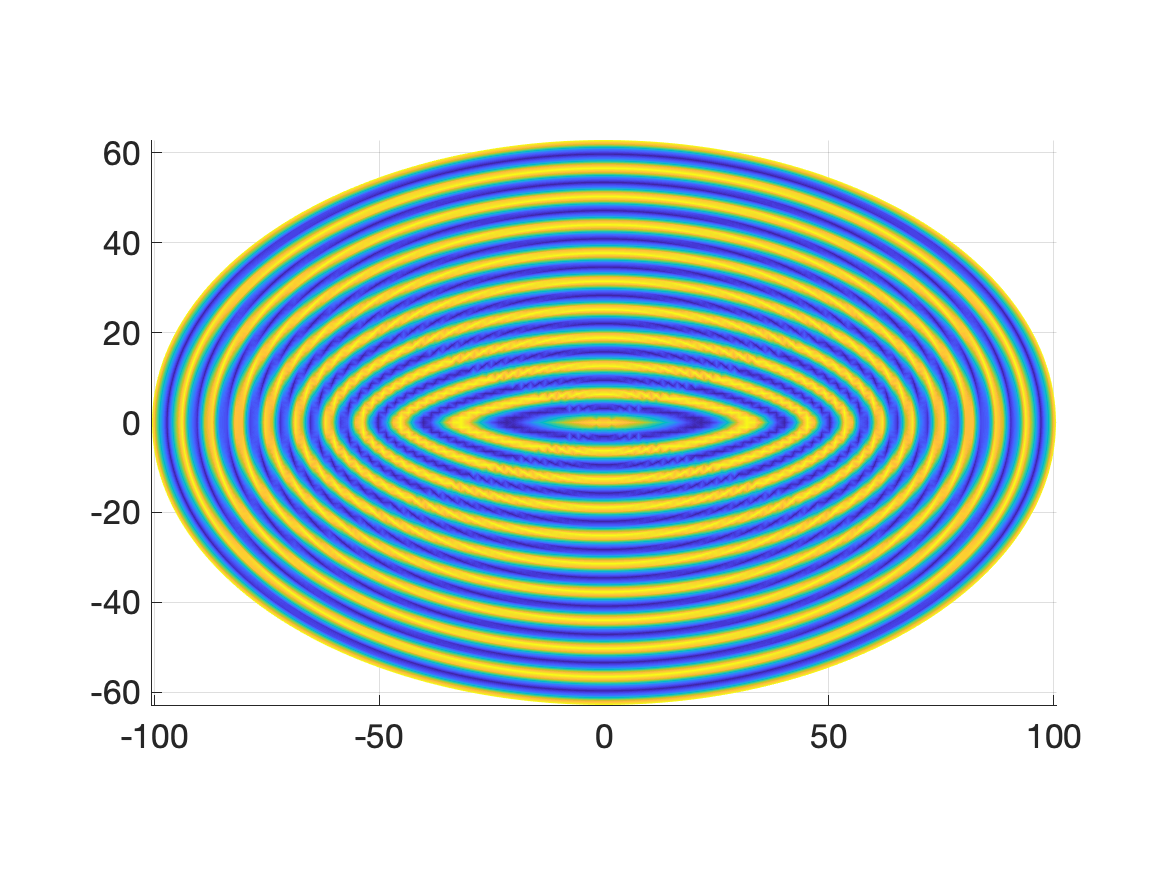} 
     \includegraphics[height=1.8in]{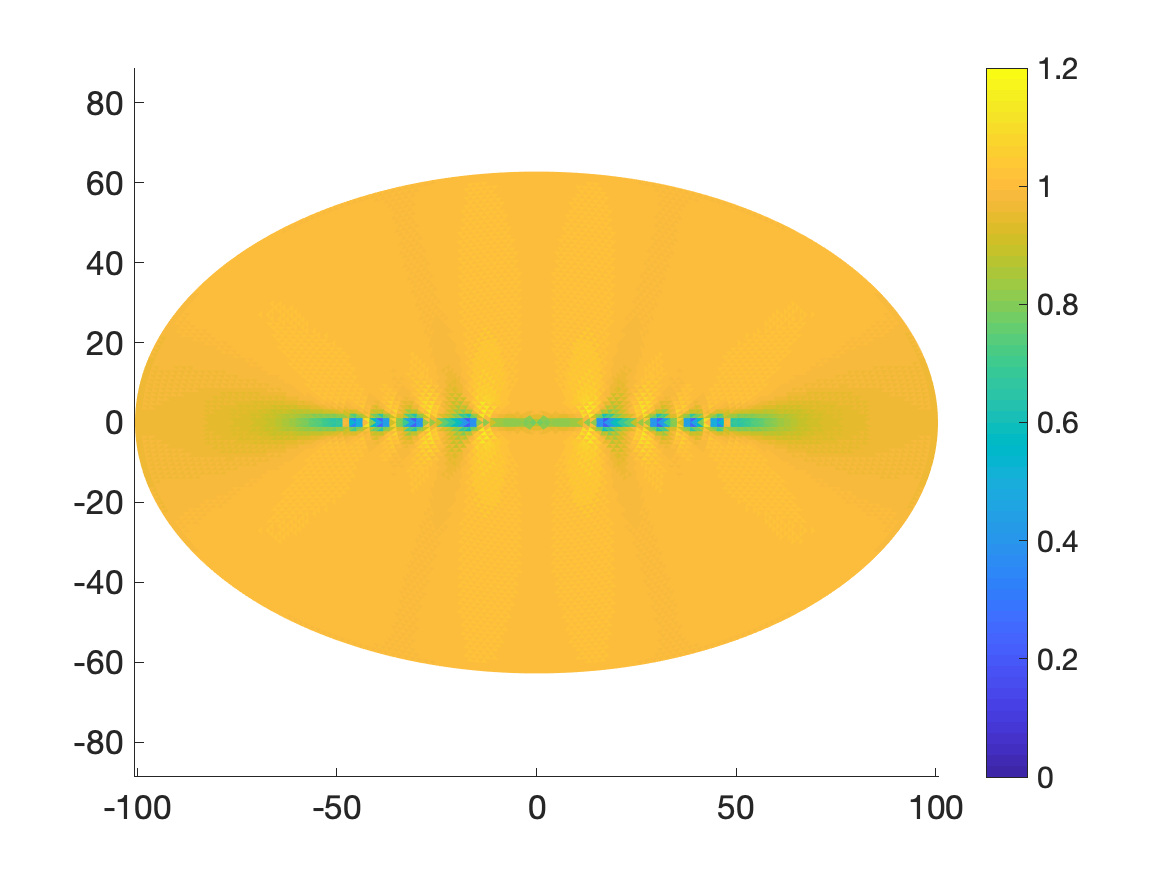} 
    \includegraphics[height=1.8in]{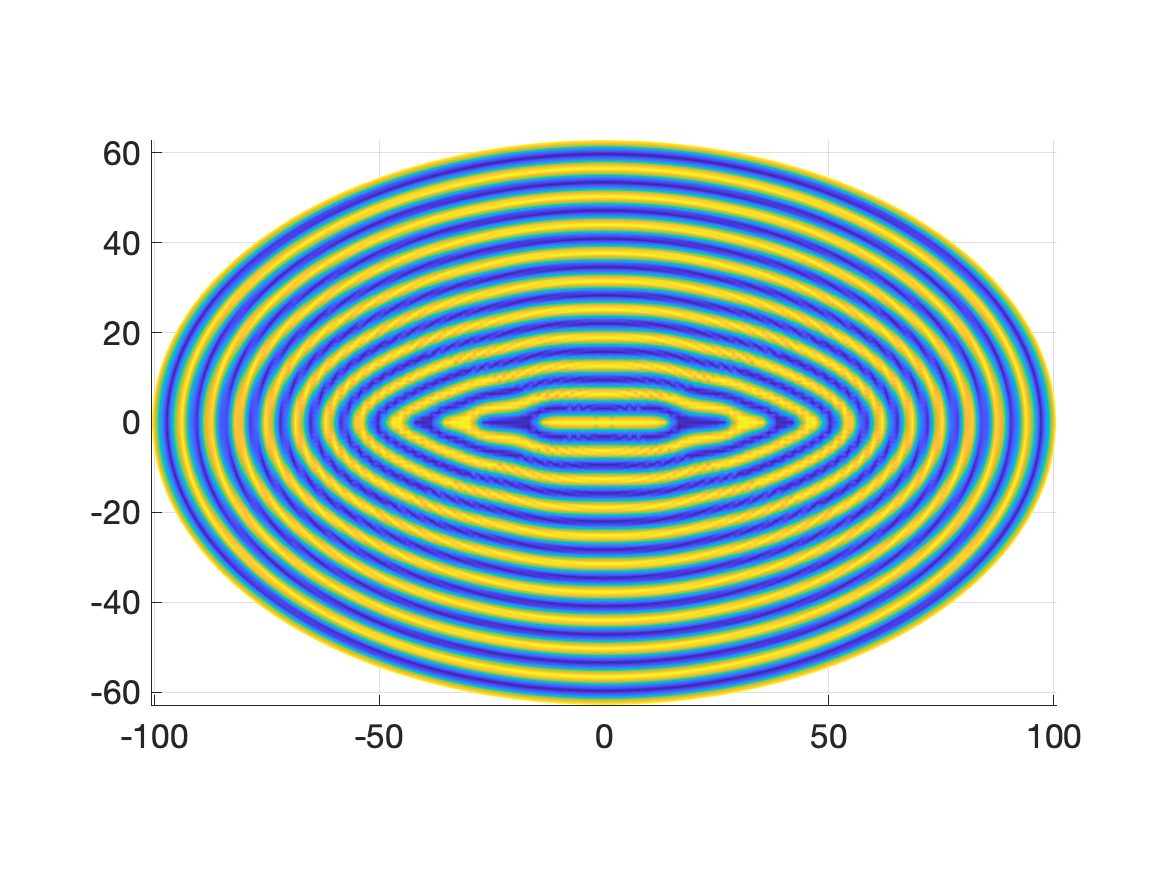}
    \caption{Stripe patterns on the elliptical domain $\left(\frac{x_1}{32 \pi}\right)^2 + \left(\frac{x_2}{20 \pi}\right)^2 < 1$. The top row is the restricted minimizer with $\mu_i = 0$ corresponding to $W^{2,2}$ phase functions and the resulting pattern has no disclinations. The bottom row is the minimizer allowing for non-orientable defects, i.e. $\mu_i$ can be non-zero. The left column displays the local wavenumber $|\nabla \theta|$ in each case,  and the right column displays the corresponding stripe pattern $h = \cos \theta$.}
    \label{fig:AG}
\end{figure}

The bottom row in Fig.~\ref{fig:AG} displays the results of iterating~\eqref{eq:SB_IMEX} without the restriction $\mu_i = 0$. There is a clear difference between the results in the two rows, due to the presence of disclinations. Note the similarity between the lower right figure and the experimental results in Fig.~\ref{ellipse-expt}. Away from the medial axis of the ellipse, we see that $|\nabla \theta| \approx 1$ even allowing for disclinations. 

As a final point, our discretization necessarily ``smears'' out the Borel measure $\mu = \mathcal{H}^1 \mres \Delta_k$ on the length scale of the mesh as we might infer from the relation $\int \varphi_i d\mu \approx \mu_i$. In particular, $\int |\sin \theta| d |\mu| \approx \sum_i \mu_i |\sin \theta_i|$ should be interpreted as a weighted average of $|\sin \theta|$ on all of the elements that intersect the support of $\mu$, giving an effective regularization as in~\eqref{eq:RSBV} with $\delta$ on the scale of the average of $|\sin \theta|$.  Since $|\nabla \theta| \approx 1$ everywhere, if we are using a triangulation with length scale $a$, we get that the average value of $|\sin \theta|$ on a band of width $a$ about the set where $\sin \theta$ vanishes is 
\begin{equation}
\label{cutoff}
\delta_{num} \simeq \frac{a}{2}
\end{equation}
This regularization is akin the the `numerical viscosity' that lead to dissipation in discretizations of the Euler equation, and vanishes as the mesh size $a \to 0$. However, for any $a > 0$, this regularization will break the gauge degeneracy between the various representatives of a phase field, and prefer the particular representative we seek, namely the one with the minimal (in $\mathcal{H}^1$ measure) jump set.

\section{Final remarks} 
\label{sec:discussion}

A mathematically convenient way to represent a stripe pattern is through a phase description. However, the phase is itself not directly directly observable. This work is an attempt to elucidate the consequences of the attendant `gauge transformations' which modify the phase but do not change any physical observable. Closed loops can now support a nontrivial ``topological charge" due to the fact that a transit along the loop can effect a non-trivial gauge transformation. This idea, with the gauge transformations coming from a (continuous) Lie group, is the cornerstone of fundamental field theories in physics. The corresponding phenomenon is less well understood, but also equally interesting, if the underlying group of transformations is discrete, as is the case in crystals, superconductors and many other condensed matter systems \cite{Kleinert_book}. 

In particular, gauge freedom within a discrete symmetry group, along with the existence of topologically nontrivial loops, naturally leads to the occurrence of defects. This is a ``soft" topological argument for the existence of defects, that relies, among other things on the existence of a ``macroscopic" phase reduction of the ``microscopic" (i.e. on the scale of the pattern wavelength) dynamics, e.g. the Swift-Hohenberg equation~\eqref{eq:sh}. It would be of significant interest to connect this argument with the considerable body of work on the rigorous bifurcation analysis of defects \cite{Haragus2012Dislocations,scheel2014small,Lloyd2017Continuation}, that start directly with the microscopic evolution equations and use tools from dynamical systems theory/center manifold reductions and functional analysis to prove the existence of defects and characterize their mutual interactions. 

One approach to studying point defects in stripe patterns is through a consideration of multi-valued fields \cite{Kleinert_book,newell1996defects,Ercolani2018phase}. A interesting alternative is through considering single-valued fields with an integrable energy density \cite{Zhang_nematics_2017,Zhang2021Computing}, that is a step towards resolving the physics of the system on a scale of the defects \cite{Pevnyi2014Modeling,Xia2021Structural}. We build on this approach by working in the class of (single-valued) SBV functions and allowing for free discontinuities to capture the defects in the underlying patterns. The key to our approach is the formulation of appropriate numerical methods for SBV functions. We propose borrowing ideas from numerical convex optimization and image processing \cite{yin2008bregman} for this purpose. 

We present an approach to variational problems in SBV that is conceptually different from the Ambrosio-Tortorelli regularization \cite{Ambrosio1992Approximation}. While our approach seems to give satisfactory numerical results, there are several associated analytic questions that we have not yet addressed. We have implicitly assumed existence of minimizers for the energies in~\eqref{eq:erelax}~and~\eqref{e:constrain}, as well as convergence of the iteration in~\eqref{eq:SB_IMEX}. These issues should be investigated rigorously. There are, however, several potential impediments for this analysis. The intuition guiding the approach in this work is the through the concept of admissible phase functions, i.e. Lipschitz functions whose level sets give a quantized measured foliation as defined by~\eqref{eq:overlap} . This is partially a topological condition. We therefore need to develop new ideas that would allow us to use uniform control of the total energy, $\mathcal{E}_{SBV}[\theta_n] \leq C$ for a minimizing sequence $\theta_n$, to guarantee that limits are also admissible phase functions. Additionally, $\mathcal{E}_{SBV}$ is closely related to the Aviles-Giga functional $\mathcal{E}_{AG}$, and the analytic challenges for the Aviles-Giga problem \cite{ambrosio1999line} are likely to have counterparts in the analysis of $\mathcal{E}_{SBV}$.  

This work is in the setting where the wavelength of the pattern is fixed relative to the size of the domain. To model situations wherein the pattern wavelength is much smaller than the domain, it would be interesting to consider the variational limits of our energy functional as the aspect ratio between the wavelength and the domain size vanishes. In this connection, it would be interesting to also consider the effect of the regularization in~\eqref{eq:RSBV}, that breaks the gauge symmetry and acts a selection mechanism. In particular, we have two small parameters $\epsilon$, the aspect ratio, and $\delta$, the regularization, and an important question is to understand how to pick their relative sizes in order to (i) obtain limit functionals through $\Gamma$ convergence and (ii) test the modeling framework by comparing the results using the limit functionals to physical observations with convection patterns and liquid crystals.

As illustrated in Fig.~\ref{fig:AG}, the numerical method in~\eqref{eq:SB_IMEX} can track topological changes in patterns, including the birth and evolution of disclinations. In this work, we have used ``uniform" meshes over the entire domain $\Omega$, as in Fig.~\ref{fig:triangulate}. Eq.~\eqref{cutoff} indicates that, in order to minimize the `numerical' regularization, we need a fine mesh to resolve disclinations and more generally, to resolve the free discontinuity $\Delta_k$, i.e. the support of the measure $\mu$. This is a `small' fraction of the domain $\Omega$, and we can get by with a much coarser mesh in the bulk of the domain $\Omega \setminus \Delta_k$. This is indeed one of the advantages of the phase reduction, as the modulation ansatz~\eqref{eq:modulation} allows us to replace a rapidly varying pattern field $\psi$ by a phase gradient $k = \nabla_X \Theta$ that varies slowly away from the defects. This provides a strong impetus to design adaptive mesh refinement schemes that resolve the free discontinuities accurately, while efficiently solving for the pattern by using coarse meshes elsewhere.

 In this work, we have restricted our considerations to 2 dimensional patterns containing point-like topological defects, motivated primarily by convection in fluids. Higher dimensional analogs, namely the foliation of 3 dimensional space by phase surfaces  with potential singularities  have been observed/computed in singular optics \cite{Dennis_Local_2004,Dennis_Polarization_2008} and liquid crystals \cite{Machon2019Aspects,chen2009symmetry,Alexander_Developed_2012,Santangelo_Curvature_2005}. The corresponding phase singularities can have a rich structure including loops and composite defects \cite{Aharoni2017Composite}, fractional defects \cite{newell2012pattern}, knots \cite{Tkalec_Knots_2011}, as well as the possibility of topological changes through reconnection \cite{Berry_Topological_2007,Berry_Reconnection_2012} of defect curves. We believe that the framework in this paper, in conjunction with extensions of our numerical methods using adaptive mesh refinement, have the potential to spur analytic as well as numerical advances in the study of evolving phase surfaces and their defects in higher dimensions.

\section*{Acknowledgements} We are grateful to Amit Acharya, Nick Ercolani, Gabriela Jaramillo, Alan Newell, Guanying Peng, Tien-Tsan Shieh and Raghav Venkatraman for useful discussions. This work was supported in part by the NSF award GCR-2020915. We are especially thankful to Irene Fonseca for her very clear lectures, in the GCR working group seminar, on SBV functions and free discontinuity problems.

\bibliographystyle{AIMS}
\bibliography{pattern_ref.bib}

\end{document}